\documentclass[11pt]{llncs}

\usepackage[ddmmyyyy]{datetime}

\usepackage{amsmath} 
\usepackage{amsthm} 
\usepackage{amssymb}
\usepackage{graphicx}
\usepackage{fullpage}
\usepackage{bm}
\usepackage{stmaryrd}
\usepackage[hidelinks]{hyperref}
\usepackage[capitalise]{cleveref}
 
\usepackage{xcolor}

\DeclareMathAlphabet{\mathpzc}{OT1}{pzc}{m}{it}

\newcommand{\mydef}[1]{\textit{#1}}

\newcommand{\blockgraph}{blockgraph}
\newcommand{\Blockgraph}{Blockgraph}
\newcommand{\initN}{\textit{init}}
\newcommand{\initB}{\initN\ node}
\newcommand{\startN}{\textit{start}}
\newcommand{\startB}{\startN\ node}
\newcommand{\updateN}{\textit{update}}
\newcommand{\updateB}{\updateN\ node}
\newcommand{\closeN}{\textit{close}}
\newcommand{\closeB}{\closeN\ node}
\newcommand{\acceptN}{\textit{accept}}
\newcommand{\acceptB}{\acceptN\ node}

\newcommand{\showPlease}[1]{#1}
\newcommand{\hidePlease}[1]{}

\newboolean{showB}  
\newboolean{showC}  
\newboolean{showBC} 

\newcommand{\mainDirLoc}{}
\setboolean{showB}{true}
\setboolean{showC}{true}

\newcommand{\hideB}[1]{\showPlease{#1}}
\newcommand{\extraB}[1]{\hidePlease{#1}}
\newcommand{\hideC}[1]{\showPlease{#1}}
\newcommand{\extraC}[1]{\hidePlease{#1}}
\newcommand{\hideD}[1]{\hidePlease{#1}}

\newcommand{\BlockgraphDefinition}{\begin{definition}[\Blockgraph]
A graph~$G$ is considered a \blockgraph\ if all of the following applies:
\begin{itemize}
\item Every node of~$G$ has to be a \startB, \acceptB, \updateB\ or \initB.
\item There must be exactly one \initB\ in~$G$.
\item The parent of every node, except the~\initB\ that has no parent, must also appear in~$G$, and it must belong to the same bank (unless the parent is the \initB).
\item For every \updateB\ in~$G$, all its referenced nodes also appear in~$G$.
\item There is an edge from each node to its parent, and edges from each \updateB\ to its referenced nodes. There must be no other edges in~$G$.
\item $G$ has no cycles.\footnote{A cycle means there is a cycle of cryptographic hash values, where each value is the hash of its previous, and we assume that producing such a cycle is impossible.}
\item Let~$v$ be a node in~$G$. We consider the (longest) sequence of nodes that starts at~$v$ and where every other node in the sequence is the parent of its predecessor. 
Ignoring the \initB\ and the \updateB s in this sequence, we must get an alternating sequence of \startB s and \acceptB s where the last node is a \startB.
\end{itemize}
\end{definition}}

\newcommand{\BlockgraphSubgraph}{\begin{lemma}
\label{lem:subgraph}
Let~$G$ be a \blockgraph, and let~$G'$ be a subgraph of~$G$.
If every node in~$G'$ has the same number of outgoing edges as in~$G$, then~$G'$ is also a \blockgraph.
\end{lemma}
\begin{proof}
All the requirements of \blockgraph s are trivially true for~$G'$, except the second and last requirements.
We start with the last requirement.
Let~$v$ be a node in~$G'$.
The parent of~$v$ also appears in~$G'$, as there is an outgoing edge from~$v$ to its parent.
But then also the parent of~$v$ has an edge to its own parent, and so on.
Thus, we get the same sequence of nodes as in~$G$, and so the requirement must hold.
Following the above case, that sequence of nodes must end at some point (as there are no cycles in~$G'$).
The only node that has no parent is the \initB.
Thus, the \initB\ also appears in~$G'$.
\end{proof}}

\newcommand{\LemmaPositive}{
Before proving \Cref{lem:positive}, we start with a technical auxiliary lemma:
\begin{lemma}
Let~$G$ be a a valid and proper \blockgraph.
We can construct a sequence of graphs with the following properties:
\begin{itemize}
\item The first graph in the sequence contains only a single node - the \initB\ of~$G$.
\item Each graph extends the previous graph in the sequence by exactly one node.
\item The last graph is equal to~$G$.
\item All the graphs in the sequence are valid and proper \blockgraph s.
\end{itemize}
\end{lemma}
\begin{proof}
Recall that \blockgraph s contain no cycles.
Thus, we can define a topological ordering on the nodes of~$G$.
For every node of~$G$, all the nodes in its subgraph come after it in the topological order.
We construct the sequence of graphs such that the $n$'th graph contains the $n$ last nodes in the topological order and the edges that connect them.
This sequence clearly answers the first three requirements.

Every graph in this sequence is indeed a \blockgraph\, according to \Cref{lem:subgraph}.
It is also valid, because if one of its nodes is invalid, then it will be invalid also in~$G$.
Assume by way of contradiction that one of these graphs,~$G'$, is not proper.
That means that there are in~$G'$ a pair of \acceptB s that don't know the \closeB s of each other, or a pair of \closeB s that don't know the \startB s of each other.
The same pair also appear in~$G$, and have the same subgraphs in~$G$ and~$G'$.
Thus, they also don't know the \closeN\ or \startN\ node of each other in~$G$ and so~$G$ is not proper, in contradiction with the assumption that it is.
\end{proof}
We shall now prove the main lemma:
\begin{proof}
Let~$G$ be a valid and proper \blockgraph.
We construct the sequence of graphs according to the above lemma.
Note that every graph in this sequence is valid and proper.
We shall prove by induction that the current lemma applies for each of these graphs, and so it applies also to~$G$ (the last graph in the sequence).
The first graph contains only the \initB\ of~$G$.
In an \initB\ all the sums of money are positive, so the lemma clearly holds for the first graph.

Assume that the lemma is true for the $n$'th graph in the list,~$G_1$.
We prove it is also true for the $(n+1)$'th graph,~$G_2$.
Note that $G_2$ extends~$G_1$ by only a single node.
If this node is a \startB, \closeB\ or \updateB, then there is no difference in the balance computation between~$G_1$ and~$G_2$ so the lemma holds for~$G_2$ as well.

If this node is an \acceptB, then we can compute the total balance according to~$G_2$ by starting from the balance of~$G_1$, 
and then apply the transactions that should be accepted according to the new \acceptB.
We denote the new \acceptB\ by~$v_a$ and its matching \startB\ by $v_s$.
Let~$t$ be a transaction that appears in~$v_s$, accepted in~$v_a$,
and was not applied in the balance computation of~$G_1$.
Note that~$t$ is one of those transactions that will be applied after we already computed the balance according to~$G_1$.
Let~$u$ be the user (the account) that transfers money according to~$t$.
As~$v_s$ is valid, the balance that~$u$ has according to the subgraph of~$v_s$ is no less than the money that should be transferred in~$t$.
We denote $v_s$'s subgraph by~$G_s$.
$u$ can ``lose'' money between $G_s$ and~$G_1$ only if it has accepted transactions in~$G_1$ that aren't accepted yet in $G_s$.
However, we claim that all the accepted transactions of $u$ in~$G_1$ are already accepted in~$G_s$.
Thus, the balance of~$u$ according to~$G_1$ is no smaller than the balance of~$u$ according to~$G_s$, and after applying~$t$ in~$G_2$ the balance of~$u$ remains non negative, as claimed.

We still need to prove that indeed all the accepted transactions of $u$ in~$G_1$ are already accepted in~$G_s$.
As~$v_s$ is valid, we can see in its subgraph ($G_s$) accepted transactions of~$u$, with all the sequence numbers up to one less than the sequence number of~$t$ (each transaction with a different sequence number).
Assume by way of contradiction that there is an accepted transaction of~$u$ in~$G_1$ that doesn't appear in~$G_s$.
If that transaction has a sequence number smaller than the sequence number of~$t$, then it is conflicting with another transaction, which is impossible as~$G_1$ is proper.
It cannot have an equal sequence number to~$t$, as it cannot be conflicting with~$t$, from the same reason as above (because~$G_2$ is proper), and it cannot be identical to~$t$, because~$t$ is first accepted in~$v_a$ (recall that~$t$ wasn't applied in the balance computation of~$G_1$).
If it has a bigger sequence number, then the \startB\ at which that transaction appears must have in its subgraph an accepted transaction with the sequence number as~$t$, which we have just proved to be impossible.
\end{proof}}

\newcommand{\LemmaProper}{\begin{proof}
Let~$G$ be a valid and proper graph, that satisfies the condition of the lemma.
I.e., more than two thirds of the voting power in every subgraph of~$G$ is in the hands of banks that are not malicious according to~$G$.
Let~$v$ be a valid node that doesn't appear in~$G$, but whose subgraph is also a subgraph of~$G$ (so we can extend~$G$ by~$v$).
Let~$G_v$ be~$G$ extended by~$v$. 

$G_v$ is non-proper if it contains a pair of either (1) \acceptB s that don't know the \closeB\ of each other, or (2) \closeB s that don't know the \startB\ of each other.
We call a pair of such \acceptB s or \closeB s an \mydef{ignorant pair}.
As~$G$ is proper, it contains no ignorant pairs.
The only way that~$G_v$ will be non-proper is if~$v$ forms an ignorant pair with one or more nodes from~$G$.
This is possible only if~$v$ is a \closeB\ or an \acceptB.
We start by assuming that~$v$ is an \acceptB.
The proof for~$v$ that is a \closeB\ is identical.

Assume by way of contradiction that~$G_v$ is non-proper.
Thus, it contains one or more ignorant pairs (where each pair includes~$v$ and another node from~$G$).
We want to find a subgraph of~$G_v$ that contains exactly one ignorant pair of nodes.
We will prove that it is impossible for the two nodes of this pair to gather the required supporting voting power, so one of them must be invalid, which means that either~$v$ or~$G$ are invalid, contradicting the assumption that they are both valid.

Considering~$G_v$, if there is only one pair of ignorant nodes, then we can keep on with the proof.
Assume otherwise that there is more than one such a pair.
We can obtain subgraphs of~$G_v$ by removing nodes from~$G_v$.
However, if we remove a node we must also remove 
all the nodes whose subgraphs contain the removed node, or otherwise we will remain with a subgraph that is not a \blockgraph.
Note that by removing nodes, the set of ignorant pairs can only shrink (new ignorant pairs cannot be created that way).

Let us consider two pairs of ignorant nodes.
Recall that~$v$ appears in both pairs.
Let~$v_1$ be the second node of the first pair, and let~$v_2$ be the second node of the second pair.
At least one of~$v_1$ or~$v_2$ must not acknowledge the other, or otherwise there will be a cycle in the graph and \blockgraph s have no cycles.
$v$ don't acknowledge any of~$v_1$ or~$v_2$, or otherwise those pairs weren't ignorant pairs.
Thus, we can safely remove (at least) one of $v_1$ and~$v_2$ (one of them that doesn't acknowledge the other) and remain with only one of the two ignorant pairs.
We will keep doing so repeatedly for every two ignorant pairs, and finally we shall remain only with a single pair, as required.

Let~$G'$ be the subgraph of~$G_v$ that we found, that contains only a single pair of ignorant nodes.
Let~$v_a$ and $v_b$ be the two nodes in this pair (note that one of them is~$v$, but we shall use~$v_a$ and~$v_b$ for convenience).
We want to prove that the sum of voting power that supports~$v_a$ and~$v_b$ is less than 4/3 of the total voting power, so they cannot both have the required voting power, and one of them must be invalid.

We shall consider three different subgraphs of~$G'$: 
(1)~$G_a$, the subgraph of~$v_a$, (2)~$G_b$, the subgraph of~$v_b$, and (3)~$G_{ab}$, the intersection of $G_a$ and $G_b$.
It can be easily seen (following \cref{lem:subgraph}) that the three subgraphs are valid and proper \blockgraph s (recall that the only ignorant pair in~$G'$ contains $v_a$ and $v_b$, and they don't appear in any of the above graphs).
We compute the voting power distribution according to each of these graphs.
According to the assumptions, more than 2/3 of the voting power in each of these distributions belongs to non-malicious banks.

We claim that a non-malicious bank might either support~$v_a$ according to~$G_a$, or support~$v_b$ according to~$G_b$, but not both.
Proof:
Assume by way of contradiction that a non-malicious bank,~$B$, supports both nodes.
This means that~$B$ has one or two \updateB s whose subgraphs contain the \closeB s of~$v_a$ and~$v_b$ (recall that we assumed, w.l.o.g, that~$v_a$ and~$v_b$ are \acceptB s).
In return,~$v_a$ and~$v_b$ have these \updateB s in their own subgraphs ($G_a$ and $G_b$ respectively).
Let $v_1$ and $v_2$ be the (earliest) \updateB s of~$B$, such that the \closeB\ of~$v_a$ appears in the subgraph of~$v_1$, and the \closeB\ of~$v_b$ appears in the subgraph of~$v_2$.
Without loss of generality, assume that $v_1$ appears before~$v_2$ in~$B$'s chain, or that $v_1=v_2$.
This means that~$v_1$ appears in $v_2$'s subgraph (or that $v_1=v_2$).
As~$v_b$ has~$v_2$ in its subgraph, that means that it also has~$v_1$ in its subgraph, and so it also has the \closeB\ of~$v_a$ in its subgraph, contradicting the assumption that $v_a$ and $v_b$ form an ignorant pair.

Recall that a voting power is ``valid'' if it belongs to a non-malicious bank, or shared by group of non-malicious banks, and ``invalid'' otherwise.
Considering $G_{ab}$'s voting power distribution, assume that the sum of valid voting power is~$x$, and the sum of non-valid voting power is~$y$. 
Note that $x+y$ is the total voting power in the system.
For simplicity, we assume that $x+y=1$.
According to the assumption we have that $y<1/3$.
The valid voting power belongs to non-malicious banks.
As we have seen, a non-malicious bank might support at most one of $v_a$ and $v_b$.
Invalid voting power, however, might belong to banks that support both nodes.
Thus, the sum of voting power that supports $v_a$ and $v_b$ according to the voting power distribution of $G_{ab}$ is at most $x+2y$.
Note that $x+2y=(x+y)+y=1+y<1+1/3=4/3$, as required.
We want to prove that this is also true when considering the correct supporting voting power -- of $v_a$ according to $G_a$, and of $v_b$ according to $G_b$.
For that cause we must consider the changes in the voting power distribution when moving from $G_{ab}$ to $G_a$ and $G_b$.
We shall prove that the sum of valid voting power in~$G_{ab}$, $x$, is still divided between the nodes, so the sum of $x+2y$ is still relevant (as an upper bound).

The major change in the voting power distribution will be due to new \acceptB s, in $G_a$ and $G_b$, that will change the total balance distribution.
Such new \acceptB s can appear only for matching \startB s that already appear in $G_{ab}$.
I.e., there cannot be a completely new block (from \startB\ to \acceptB) that appears in $G_a$ or $G_b$.
Proof:
Assume by way of contradiction that such a new block appears in~$G_a$.
This means that a new \acceptB, $v'$, appears in the subgraph of~$v_a$ together with a matching \startB, and where that matching \startB\ doesn't appear in the subgraph of~$v_b$ (or else it would also appeared in~$G_{ab}$).
There must be also a matching \closeB\ that appears in the subgraph of~$v_a$ (as~$v'$ appears in it) but not in the subgraph of~$v_b$ (or else the matching \startB\ would also appear there).
Because the only ignorant pair in~$G'$ are $v_a$ and~$v_b$, it follows that 
$v'$ must know $v_b$'s \closeB\ or else it will form another ignorant pair together with~$v_b$, as~$v_b$ doesn't have the \closeB\ of~$v'$ in its own subgraph. 
As~$v_a$ knows~$v'$, it will also know $v_b$'s \closeB,
contradicting
the original assumption that~$v_a$ and~$v_b$ are ignorant pairs.
Thus, a new \acceptB\ can appear only if its \startB\ already appears in $G_{ab}$.

Except for new \acceptB s (for existing \startB s), there might be also new \startB s without matching \acceptB s (but possibly with matching \closeB s, in case $v_a$ and $v_b$ are indeed \acceptB s\footnote{If $v_a$ and $v_b$ are \closeB s, then from the same reasoning as before, we will get that a new matching pair of \startB\ and \closeB\ cannot appear in~$G_a$ or~$G_b$, or else $v_a$ and $v_b$ wouldn't be ignorant pair.}).
The only effect that such new open blocks might have on the voting power distribution is that existing amounts of voting power will be shared with additional banks because of new uncertain transactions.
Such sharing cannot increase the supporting voting power of $v_a$ or $v_b$, because if that voting power didn't support them in the first place, then sharing it with additional bank won't change it.
However, if that voting power supported them beforehand, then now it might not support them, in case the new bank doesn't support them.
Thus, such new \startB s might only decrease the supporting voting power of the nodes.
As we are interested in the maximal sum of supporting power, we will ignore such new \startB s, and consider only the effect of new \acceptB s (for existing \startB s).

Recall that when computing the voting power distribution, we start with computing the balance of each user.
It can be seen as if each user has amount of voting power that equals to his balance, and each user delegates his voting power to his bank.
However, if the graph contains uncertain transactions that belong to a specific user, then the bank of that user has to share some of this voting power with other banks that might receive that power according to 
those transactions. 
To conclude, we can divide the total voting power according to the amount that each user contributes (the user's balance), and further divide each such amount to shared and non-shared voting power.

We are interested in voting power that was valid in~$G_{ab}$.
Such valid voting power is based on balances of users that belong to non-malicious banks (the voting power of other users is delegated to malicious banks, and so it is invalid).
Even if the user belongs to a non-malicious bank, the shared voting power that he contributes will be invalid in case one of the banks that share it is malicious.
We shall consider the valid voting power that each user (that sits in non-malicious bank) contributes in~$G_{ab}$, and we shall see what happens with this voting power in~$G_a$ and~$G_b$.
Recall that we defined~$x$ to be the sum of valid voting power in~$G_{ab}$.
$x$ is exactly the sum of all the contributions of valid power from those users.

Let~$u$ be a user of a non-malicious bank~$B$.
Concerning~$B$, we have already proved that it might either support (exactly) one of~$v_a$ or $v_b$, or not support any of them.
Assume that it supports~$v_a$.
In this case it cannot have a new \acceptB\ that appears in~$G_b$.
Proof:
Assume by way of contradiction that a new \acceptB\ of~$B$ does appear in~$G_b$.
Denote that new \acceptB\ by~$v_1$.
That means that~$v_1$ appears in the subgraph of~$v_b$.
It cannot appear also in~$G_a$, as otherwise it would appeared also in~$G_{ab}$, and then it won't be ``new''.
As~$B$ supports~$v_a$, according to~$G_a$, then one of its nodes in~$G_a$ has the matching \closeB\ of~$v_a$ 
in its subgraph.
As this node appears in~$G_a$ while~$v_1$ does not, then it must come before~$v_1$ in~$B$'s chain (recall that~$B$ is non-malicious, so it has no forks).
Thus,~$v_1$ has the matching \closeB\ of~$v_a$ in its subgraph, and so does~$v_b$, contradicting the assumption that $v_a$ and $v_b$ are ignorant pair.

Recall that we are currently considering what happens with the valid voting power in~$G_{ab}$ that contributed a user~$u$ (in bank~$B$).
Assume, once again, that~$B$ supports~$v_a$.
This means that the non-shared part of the voting power that~$u$ contributes supports~$v_a$, and its shared voting power either supports~$v_a$ or not, depending on the identities of the other banks.
The voting power of~$u$ cannot support~$v_b$, because we have already seen that~$B$ cannot support both $v_a$ and $v_b$.
The question is if the voting power that belonged to~$u$ in~$G_{ab}$ might be later used to support~$v_b$ in~$G_b$.
This is possible only if a transaction of~$u$ is accepted in~$G_b$ (and not in~$G_{ab}$) such that some of its money is transferred to another bank, that does support~$v_b$.
Denote that transaction by~$t$.
We proved above that a new \acceptB\ of~$B$ cannot appear in~$G_b$.
Thus, if~$t$ is accepted in~$G_b$ and transfers money outside of~$B$, then it must be an uncertain transaction in~$G_{ab}$. 
This means that the voting power that will be transferred is shared, according to~$G_{ab}$.
If that voting power goes to a malicious bank, then this power was invalid in~$G_{ab}$ (as it was shared with a malicious bank), and is not of our concern.
Assume otherwise that it goes to a non-malicious bank,~$B'$.
We want this voting power to support~$v_b$.
This means that~$B'$ should support~$v_b$.
As~$B'$ is non-malicious, it cannot also support~$v_a$.
As~$t$ is an uncertain transaction in~$G_{ab}$ that attempts to transfers money to~$B'$, the corresponding voting power (in~$G_{ab}$) does not support~$v_a$.
The question is if that voting power still won't support~$v_a$ when considering the voting power distribution at~$G_a$.
The only way that voting power might indeed support~$v_a$ is if~$t$ will cease to be ``uncertain'' in~$G_a$ (but not by being accepted), so the corresponding amount of money will be no longer shared with~$B'$.
This can be either because~$t$ was rejected at a new \acceptB, or otherwise that another transaction with the same sequence number was accepted.
We shall now prove that both options are impossible.

We start with contradicting the first case, where~$t$ is rejected in~$G_a$.
Note that~$t$ has to be both accepted in~$G_b$ and rejected in~$G_a$ by two conflicting \acceptB s for the same \startB\ at which~$t$ appears in~$G_{ab}$.
Clearly that means that~$t$ appears in a malicious bank's block.
For it to be both accepted and rejected, there must be also two conflicting \closeB s, where one of them sees another conflicting transaction in another block, and the other does not.
The result is that we have two conflicting \closeB s with two matching (and conflicting) \acceptB s for the same bank.
They all appear in the union of~$G_a$ and $G_b$, and so they also appear in~$G'$.
Each of the \acceptB s cannot have in its subgraph the \closeB\ of the other, as this results in invalid graph (a malicious bank cannot identify its own malice).
Thus, they form a pair of ignorant nodes in~$G'$.
The only ignorant nodes in~$G'$ are~$v_a$ and~$v_b$.
However, those two conflicting \acceptB s cannot be~$v_a$ and~$v_b$, as they appear in~$G_a$ and~$G_b$ which are the subgraphs of~$v_a$ and~$v_b$ respectively, and according to our definition of subgraph, a node does not form part of its own subgraph.
Thus, it cannot be that~$t$ is accepted in~$G_b$ and rejected in~$G_a$.

The second case is where another transaction of~$u$, different than~$t$ but with the same sequence number, is accepted in~$G_a$. 
As~$t$ gets accepted in~$G_b$, we have two conflicting transactions that get accepted in~$G'$ (which contains both~$G_a$ and~$G_b$).
This means that the two \startB s that contain those conflicting transactions have a matching pair of ignorant \closeB s (or else at least one of the transactions wouldn't be accepted). 
However, the only ignorant pair in~$G'$ are~$v_a$ and~$v_b$, and from the same reasoning above, those ignorant \closeB s cannot be~$v_a$ and~$v_b$ themselves.

Recall that we considered above the valid voting power of a user~($u$) whose (non-malicious) bank~($B$) supports~$v_a$.
We have seen above that if some of its voting power will support~$v_b$ according to~$G_b$, then it won't support~$v_a$ in~$G_a$.

The next case is if~$B$ supports neither $v_a$ nor $v_b$.
In this case the voting power of~$u$ in~$G_{ab}$ couldn't support~$v_a$ nor $v_b$, as~$B$ doesn't support them.
The question is if that voting power might later support both of them~--~$v_a$ in~$G_a$ and~$v_b$ in~$G_b$.
This is only possible if~$u$ has a transaction that transfers money to a bank,~$B''$, that supports both.
This is true because (1) two conflicting transactions cannot be accepted, as it requires an ignorant pair of \closeB s, and the only such pair that exist is the pair of~$v_a$ and~$v_b$, and (2) two consequent transactions of the same user cannot be accepted, because the \startB\ that contains the later transaction must know the \acceptB\ that accepted the first transaction, and so the block that accepts the second transaction is a new block in~$G_a$ or~$G_b$ (both its \startB\ and \acceptB\ don't appear in~$G_{ab}$), which we proved to be impossible.
If that~$B''$ supports both~$v_a$ and~$v_b$, then it is malicious (as non malicious banks can support at most one of them).
If the transaction that transfers the money to~$B''$ does not appear in~$B$, then it is an uncertain transaction in~$G_{ab}$, and then that voting power is considered invalid in~$G_{ab}$, and is not of our concern.
Assume otherwise, that this transaction appears in a block of~$B$.
For this voting power to support both~$v_a$ according to~$G_a$ and~$v_b$ according to~$G_b$, a corresponding \acceptB\ must appear both in~$G_a$ and in~$G_b$.
As~$B$ is not malicious, it can have only one such \acceptB, that can appear at most in one of~$G_a$ or $G_b$ (or else it would appear also in~$G_{ab}$).
Thus, at most one of~$v_a$ or $v_b$ can be supported by this voting power.

The result is that the valid voting power from~$G_{ab}$ of every user that sits in non-malicious bank can support at most one of $v_a$ or~$v_b$ (in~$G_a$ and~$G_b$ respectively).
This is what we claimed, and so the total supporting voting power is less than 4/3, and at least one of~$v_a$ and~$v_b$ cannot be valid.
\end{proof}}

\newcommand{\Protocol}{
\begin{itemize}
\item When receiving a transaction from a user, check if it may be rejected, according to the rejection restriction requirement (see \Cref{sec:model}). If it may be rejected, \mydef{reject} it. Otherwise, add it to a list of waiting user transactions.
\item If you can create a \startB\ (your previous \startB\ has a matching \acceptB) and your list of waiting transactions is not empty, then extract the transactions from the list, keep the ones that are still valid according to your current graph (i.e., that cannot be rejected according to the rejection restriction), and create a new \startB\ from those transactions (as long that at least one of the transactions was valid indeed). \mydef{Reject} the transactions that aren't valid anymore.
\item If you have an open transactions block (a \startB\ with no matching \closeB, or a \closeB\ with no matching \acceptB) and you can validly create an \acceptN/\closeB, then create it. When creating a \closeB, \mydef{Reject} transactions that appear in the current \startB\ that have conflicting transactions according to the current \closeB's subgraph.
\item When receiving a node from another bank, check if it is already included in your \blockgraph.
If not, then check if the nodes it references are included in your \blockgraph.
If not, ask the other banks for these nodes.
If (or once) your \blockgraph\ includes all the nodes that are referenced by the node you received (but doesn't include that node itself), then make sure that (1) the node is valid, and (2) that adding it to your graph won't make your graph improper.
If it satisfies both criteria, then add the node to your \blockgraph, and create a new \updateB\ that references it.\footnote{In a practical implementation we don't have to respond with an \updateB\ if this won't deliver any important information for the protocol. E.g., there is no importance to acknowledge an acknowledgment to a previous acknowledgment of your own node.}
\item After creating a node in one of the above cases, send this node to all other banks.
\item After adding an \acceptB\ to the graph, \mydef{Accept} the transactions that should be applied according to the graph.
\end{itemize}
We want to prove that executions of our system fulfill the requirements we defined in \Cref{sec:model} (\textbf{Agreement}, \textbf{Positive Balance}, \textbf{Termination} and \textbf{Rejection Restriction}).
We start with the \mydef{rejection restriction} requirement.
According to the protocol, the transaction~$t$ that is submitted to a valid bank~$B$ will be rejected if~(1)~$B$ is allowed to reject it according to the rejection restriction requirement, or (2) if it appears in a \startB\ of~$B$, and according to its matching \closeB\ there is another conflicting transaction (a different transaction of the same account with the same sequence number). 
In both cases the rejection restriction is not violated.

The positive balance requirement holds by \cref{lem:positive} because the graph of every valid bank is always proper (as we don't accept nodes that cause it to be improper).

Agreement holds if every \acceptB\ that is accepted by one valid bank, will be accepted by all the other valid banks.
Such \acceptB\ must be valid, or else it wouldn't have been accepted by the first valid bank.
Thus, it won't be accepted by another bank only if it will make its graph improper.
If the voting power distribution according to every valid bank's \blockgraph, at every time point, is such that more than two thirds of the voting power is valid (i.e., in the hands of valid banks), then adding a new, single, valid node won't make any of these graphs improper, following \cref{lem:proper}.

After agreement we left only with the termination requirement.
While practically termination should pose no concern, it is theoretically problematic.
We call it \mydef{the problem of migrating power}, and it generally applies to every cryptocurrency system where the possible number of admins is unbounded and where some resource is used as the base for decisions.
The problem with resources is that they might migrate. 
One day they are in one hand, the second they are in another.
If the lazy user sends his transaction to an admin that already lost its power, 
then that admin will have to forward it to the one that is now in power (that holds a big share of resources). 
But, what if by the time this transaction reaches the second admin, the power has already migrated to another admin?
More generally, if the time it takes the power to migrate is smaller than the delay of messages from those admins that have already lost the power (not very reasonable), then a user that sent his transaction to an ``old'' admin will never get a response.

In practice this problem doesn't seem to make sense, and anyway a user that sees that his transaction doesn't get to the right locations can resubmit his transaction to the current power owners.
Power migration is problematic only if the power migrates infinitely, going through infinite number of admins.
Now, let's be realistic, our coin won't live forever.
In that finite time that it does, there is a finite set of banks that will ever had any voting power, and in particular a finite set of valid banks.
Thus, if we assumed before that two thirds of the voting power is always in the hands of valid banks, then we can further assume that there is a finite set of valid banks that together always hold more than two thirds of the voting power.
Following this assumption, we shall prove that termination holds.

We start with the following claim:
Let~$B$ be a valid bank, that created a \startB, but didn't create yet a matching \closeB.
We claim that eventually~$B$ will create the matching \closeB.
The same claim (with the same proof) is that if~$B$ created a \closeB\ with no matching \acceptB, then eventually it will create that matching \acceptB.

Proof:
According to the protocol, $B$ will create the \closeB\ at the moment that it can.
It can create such node only if the sum of voting power that supports the matching \startB\ is at least two thirds of the total voting power.
When~$B$ created the \startB, it sent it to all other banks.
Now, recall that by the assumption there is a finite set of banks that more than two thirds of the voting power is always in their hands.
We denote this set by~$S$.
The banks in~$S$ will eventually accept that \startB\ of~$B$ (once they accepted the nodes that it references), as it is valid, and it won't make their graph improper (as we have seen before).
After they accept it, they will create corresponding \updateB s, that reference this \startB, and they will send their \updateB s back to~$B$.
In return, $B$~will create new \updateB s that reference the \updateB s it received from the banks of~$S$.
At this point~$B$ has the support of all the banks from~$S$.
Considering the voting power distribution, more than two thirds of the voting power is in their hands, so~$B$ has the required supporting voting power, and it will create a \closeB.

Now we can easily prove that the termination requirement holds.
Proof:
Let~$B$ be a valid bank that receives a transaction from its user.
According to the protocol,~$B$ will put the transaction in its next \startB, or otherwise it will reject it (if it is non valid).
The only thing that might prevent~$B$ from creating its next \startB\ is if it already has an open \startB\ without a matching \closeB\ or \acceptB.
According to the above claim,~$B$ will eventually create matching \closeB\ and \acceptB, and then it will also create a new \startB\ that will contain the required transaction.
Once again, according to the above claim,~$B$ will eventually create matching \closeB\ and \acceptB\ to the new \startB\ as well.
Once it created the \acceptB, it also either accepts or rejects the transaction.

To summarize, we proved above the following claim:
\begin{lemma}
Let $\mathbb{B}$ be the set of possible public keys of banks, let $r$ be an execution of our system and let $\mathbb{B}'\subseteq\mathbb{B}$ be the subset of valid banks in~$r$.
If there is a finite set  $\mathbb{B}''\subseteq\mathbb{B'}$ such that at every time point the voting power distribution according to the \blockgraph\ of every $B\in\mathbb{B'}$ is such that more than $2/3$ of the voting power is in the hands of banks from~$\mathbb{B''}$, then~$r$ satisfies the requirements defined in \Cref{sec:model}.
\end{lemma}

\subsection{Implications}
The protocol and settings defined above are rather limited, and intended to provide only the minimum that is required in order to have a cryptocurrency that achieves the requirements we defined under asynchronous communications.
We provide here some details and implications we ignored.

The first case concerns allegedly malicious users.
According to the \blockgraph\ definition, a user that submitted conflicting transactions might not be able to submit any additional transaction, which makes its money unusable.
This can be seen as a punishment to such user, but it might be a too harsh punishment, as it might be an honest mistake, with no malicious intentions.
To overcome it, we can allow a user to submit a ``group of transactions'' with a single sequence number.
Such group won't be considered as conflicting with any subset of the transactions it contains.
 
Another concept that wasn't discussed above is the commission.
We have mentioned in the introduction that a user will pay commission for each of his issued transactions.
The simplest approach is that the commission will be a fixed percentage of the transaction sum, that is given to the bank of the user that issued the transaction. 
Another option is to divide the commission between the user's bank (probably about an half of the commission) and the banks that supported the block at which that transaction appeared (as they all took part in the agreement process). 
If we want to encourage participation of as many banks as possible, then we can define the commission percentage to be variable, where this percentage gets bigger as more banks (voting power) support the block.
As the issuing bank receives a fixed share of the commission, the commission it receives will be bigger as more banks support its block.
This incentivizes the issuing bank to ask for the acceptance of as many banks as it can.

However, sharing the commission between the banks that supported the block might cause problems if the same transaction appears in blocks of different banks, as it is not clear then according to which block we divide the commission.
Such case is possible if the user submitted a transaction to his bank and got no response, so he sent it also to another bank, and as a result both banks might create blocks that contain this transaction.
In order to solve this problem we can define that the commission distribution is computed only according the block of the original bank of the user.
If that bank stops responding before it manages to accept its block, then the amount of commission that should have been taken remains effectively as inaccessible money.}

\newcommand{\Execution}{
Let $(A,I,\mathbb{B},P)$ be a cryptocurrency system.
An execution of the system is defined by the tuple $(X,Users,Admins,F,Tx,Msgs,Ac,Re)$.
$X$ is a (possibly infinite) set of nodes that act as users and admins.
$Users$ is a map $X\to2^{A}$ that assigns each node a set of account numbers, and $Admins$ is a map $X\to2^{\mathbb{B}}$ that assigns each node a set of admin IDs.
For every $x_1,x_2\in X$, if $x_1\neq x_2$ then $Admins(x_1)$ and $Admins(x_2)$ are disjoint sets.
If $a\in Users(x_1)$, then we say that~$x_1$ represents the account number~$a$.
The same goes with admins.
A node that represents an account number 
is considered a \mydef{user}.
A node that represents an admin 
is considered an \mydef{admin}.
Note that a node can be both a user and an admin.
The nodes communicate by message passing.
We assume a fully connected network.
Users can create transactions and send them to the admins.
The admins can send messages between them and accept/reject user transactions.
Admins should send messages and accept/reject transactions only if they should do so by the protocol.
Admins that don't do so are considered malicious.
Malicious admins can perform arbitrary operations, but they cannot send messages that contain data that was digitally signed by someone else, unless they received that signed data beforehand. 
The same goes with accepting/rejecting transactions.

$F$ is a map $X\to\{\mathbb{R}^{\geq0}\cup\infty\}$ that defines for every $x\in X$ the time when~$x$ fails by crashing.\footnote{A more practical definition should include for every node the time when it becomes active. We avoided doing so for ease of exposition.}
If the time value is ``$\infty$'', the node never crashes.
A node that crashed cannot create transactions, send/receive messages or accept/reject transaction after it crashed.
An admin is considered valid if it is not malicious and if its representing node doesn't crash.

$Tx=\{(tx,b,t)\}$ is the set of transactions that were submitted by the users, where $tx$ is a transaction (as defined in \Cref{sec:model}), $b\in\mathbb{B}$ is the admin to whom the transaction was submitted, and $t\in\{\mathbb{R}^+\cup\bot\}$ is the time that the node that represents~$b$ received~$tx$, where~$\bot$ means that it never receives it.
$t$ can be $\bot$ only if either the node that represents~$b$ or some node that represents the source account of~$tx$ crashes.
If~$t\neq\bot$, then both the node that represents~$b$ and some node that represents~$a$ must not crash before~$t$.

$Msgs=\{(b_s,b_t,m,tSend,tReceive)\}$ is the set of messages that are sent between admins (or, more precisely, between nodes that represent admins).
$b_s,b_t\in\mathbb{B}$ are the source and destination admins respectively, both must be represented by some nodes. 
$m$ is the message contents, ${tSend\in\mathbb{R}^+}$ is the time the message was sent, and $tReceive\in\{\mathbb{R}^+\cup\bot\}$ is the time the message was received.
Assume that $x_s,x_t\in X$ are the nodes that represent $b_s$ and $b_t$ respectively.
If~$x_s$ crashes at time~$t'$, then $tSend<t'$.
If $tReceive=\bot$ this means that the message was never received, which is possible only if~$x_s$ or~$x_t$ crashes.
Otherwise, $tReceive>tSend$ and if~$x_t$ crashes at time~$t'$ then $t'>tReceive$.

$Ac=\{(b,tx,t)\}$ is the set of accepted transactions, where $b\in\mathbb{B}$ is the admin, $tx$ is a transaction that was submitted by a client and $t\in\mathbb{R}^+$ is the time that~$b$ \mydef{accepted}~$tx$.
$Re$ is a similar set that describes the \mydef{rejected} transactions.

Using the above definition, we can examine if a given execution satisfies the requirements defined in \Cref{sec:model}.
}

\begin{document}
\title{Cryptocurrency with Fully Asynchronous Communication based on Banks and Democracy} 
\author{Asa Dan}
\institute{}
\maketitle
\centerline{\email{asa.dan.te@gmail.com}}~\\
\centerline{\large{\today}}
\begin{abstract}
Cryptocurrencies came to the world in the recent decade and attempted to offer a new order where the financial system is not governed by a centralized entity, and where you have complete control over your account without the need to trust strangers (governments and banks above all).
However, cryptocurrency systems face many challenges that prevent them from being used as an everyday coin.
In this paper we attempt to take one step forward by introducing a cryptocurrency system that has many important properties.
Perhaps the most revolutionary property is its deterministic operation over a fully asynchronous communication network, which has sometimes been mistakenly considered to be impossible.
By avoiding any temporal assumptions, we get a system that is robust against arbitrary delays in the network, and whose latency is~only a function of the actual communication delay. 
%
The presented system is based on familiar concepts~--~banking and democracy.
Our banks, just like normal banks, keep their clients' money and perform their clients' requests.
However, because of the cryptographic scheme, your bank cannot do anything in your account without your permission and its entire operation is transparent so you don't have to trust it blindly.
The democracy means that every operation performed by the banks (e.g., committing a client transaction) has to be accepted by a majority of the coin holders, in a way that resembles representative democracy where the banks are the representatives and where each client implicitly delegates his voting power (the sum of money in his account) to his bank.
A client can switch banks at any moment, by simply applying a corresponding request to the new bank of his choice. 
\quad The presented approach employs the advantages of centralization while still providing a completely trustless and decentralized system. 
By employing concepts from everyday life and attaining high throughput and low latency for committing transactions, 
the hope is that this paper will lay the foundations for a cryptocurrency that can be truly used as a practical coin on a daily basis. 
\end{abstract}




\section{Introduction}\label{sec:intro}
%
%
In its most basic form, cryptocurrency is a digital coin 
accompanied by cryptographic tools that provide several benefits.
In a digital coin, your money is just a number (the amount of money you have). 
In today's world, these numbers are usually kept by the banks.
As a result, you have to trust your bank, as it is the only entity that can tell how much money you have (and whom you must address in order to use this money).
\hideD{If your bank suddenly decides that you have half the money you had before (maybe because the government told him to do so), you are in problem.
}On the contrary, in the common cryptocurrency scheme the idea is to deploy a decentralized and trustless system where there is no single entity that can be trusted to tell you how much money you (or someone else) have.
Instead, the amount of money you have should be determined by consensus among the other users of the cryptocurrency -- everyone should agree, somehow, that you have that specific amount of money.
Moreover, this vague consensus should be reached regarding all the transactions (money transfers) that are committed.
This is required in order to prevent \mydef{double spending}, where one can use the same coin more than once, 
without its recipients knowing this.
Once there is a consensus on a payment, the corresponding recipient can be sure that he received the money (because everyone agrees that the money is his).
\hideC{A second payment with the same coin will not be accepted, as the corresponding (second) recipient will see that the payment is not in the consensus.
}The big question is how can you reach such consensus in a setting where you can trust no one, and where it is not generally defined who the other players are.
This is also known as  the ``Consensus in the Permissionless Model'' problem \cite{pass2017hybrid}.

Bitcoin emerged in~2008~\cite{bitcoin}
offering a solution to this problem and putting the cryptocurrency in the headlines.
However, 
\hideC{despite its glorious success (for the moment), }Bitcoin has many known flaws.
Among its flaws are the time it takes a transaction to be accepted~\cite{byzcoin}, its limited transactions throughput~\cite{lightning}, and finally, the great energy consumption involved in keeping it alive. 
There are numerous works and other cryptocoins that are trying to fix these flaws. 
\hideC{Especially, there are strong environmental and economical reasons to alleviate its energy consumption that follows from its innovative mechanism for solving the ``Consensus in the Permissionless Model'' problem.}\hideB{ More information concerning Bitcoin and other protocols appears in \Cref{sec:related}.

}The common denominator for probably all of the numerous different cryptocoins is that there are roughly two groups -- users and administrators.
The users are the simple persons/clients that want to hold coins and use them for any type of trade or investment. 
\hideB{Each user has a matching pair of secret and public keys, where the public key is the user's ``account number'', and the secret key is used by the user for creating digital signatures.
}The administrators' role is to make sure that there is consensus concerning the current balance of the user accounts and concerning the committed user transactions.
An admin can be represented by a matching pair of secret and public keys, just like the users.
The set of admins might be well defined in advance, in which case we say it is a \mydef{permissioned} settings.
Otherwise, everyone can be an admin (a \mydef{permissionless} settings).
\hideB{The administrators must invest resources such as computation, storage and bandwidth in order to manage the consensus, so it might be not suitable to simple users.
As a revenue for their efforts, the admins receive coins -- either by commission from transactions, or by creating money from thin air, i.e., stamping new money.
\hideC{The coin they receive plays a double role -- it provides both revenue for their investment, and an incentive to keep the stability of the coin, as otherwise the coin's value will drop and the administrators' true gain will decrease as well (the coins they will receive shall have lower value).

}}When a user wants to transfer money from his account to another, he creates a transaction that contains the required information, \hideB{attaches to it a matching digital signature, }and sends it to one or more administrators.
The admins should make sure that the transfer is ``legal'', and if it does, they accept it.
Accepting a transaction means that the money transfer described by the transaction was executed\hideC{ (so, from now on, the money belongs to the recipient)}.
A transaction that was accepted by one admin should be eventually accepted by all the admins. 
By applying the set of accepted transactions, we can compute how much money each user holds -- the sum of money he received minus the sum of money he transferred.
Clearly that balance must be non-negative. 
In \Cref{sec:model} we introduce a formal, and reasonable, definition for the cryptocurrency problem.

Following the above discussion, there must be consensus regarding the accepted transactions.
We should be careful, however, when using the word \textit{consensus}.
In distributed computing theory, the consensus problem \cite{benor} 
involves a set of nodes, each holding an input value, and the nodes must all agree on a \textbf{single} input value (an input value of one of the nodes).
This is similar to our case in cryptocurrency, where finally all nodes (admins) should agree which transactions have been accepted (where a user might send conflicting transactions to different nodes).
Many works on cryptocurrency (e.g. \cite{hashgraph,tendermint,tschorsch2016bitcoin,vukolic2015quest}) mention the FLP impossibilty result \cite{FLP} that states that consensus cannot be solved when the communication channels between the nodes are asynchronous, and where nodes might fail (by stopping to respond).
\hideB{A communication channel is considered asynchronous if a message that is sent over the channel can suffer an arbitrary delay. 
}Following the FLP result, it might be deduced that a (deterministic) cryptocurrency system cannot assume asynchronous communications.
The first immediate result presented in this paper shows that this is \textbf{wrong}. 
In \Cref{sec:solution} we define a deterministic system that solves the cryptocurrency problem that is defined in \Cref{sec:model}, and does so over asynchronous communication network.
It follows that the cryptocurrency problem is not as hard as the classical consensus problem.
The crucial difference between the two problems is in the way we deal with conflicting values/transactions.
In the classical consensus problem, if there is a set of conflicting values, the nodes must accept on a single value out of the conflicting ones.
In cryptocurrency, transactions are conflicting when a user attempts to double spend -- pay the same coin twice.
In such a case, we claim that we don't have to accept one of the conflicting transactions -- we can reject them both.
Rejecting both conflicting transactions makes sense, as such conflicting transactions result from invalid behavior of a user.
Irresponsible users should be aware that their transactions might be rejected.

In order to take advantage of the fact that cryptocurrency is easier to solve than consensus, we cannot use \mydef{blockchain} based solutions.
A blokchain 
is basically an ordered list of all the accepted transactions.
The order of the list is accepted by everyone.
Agreeing on such a total order is as hard as solving the classical consensus problem \cite{chandra1996unreliable}, and so it is unsolvable under asynchronous communications according to the FLP.
While convenient, total order of the transactions is not necessary.
For example, if Alice transfers \$20 to Charlie, and Bob transfers \$30 to Charlie, there is no importance which transfer comes first.
\hideD{The order is important only for the outgoing transactions of a single user, to make sure that the user doesn't pay the same coin twice.}

Note that most if not all the existing cryptocoins operate using asynchronous communication channels.
However, none of them is known to meet the requirements we define\hideB{ (and particularly those that are blockchain based cannot meet the requirements)}.
More precisely, many of the existing cryptocoins operate in non determinisitc fashion, and so they can meet the requirements we define with high probablity.
However, they often trade latency for safety \hideC{-- the probability for success goes higher as we wait more time before accepting a transaction }\cite{bitcoin}.\hideC{
Such latency is clearly unwanted.}

As the cryptocurrency problem involves some sort of (weaker) consensus, and consensus is roughly achieved when the majority accepts on something, we need to find a way to define majority.
Algorithms that operate in permissioned settings are usually based on the number of existing nodes.
E.g., if more than two thirds of the nodes accepts a transaction, then it can be considered accepted.
In permissionless settings, the number of nodes is not known and is not relevant, as everyone can produce arbitrary number of nodes.
Instead, a resource that cannot be easily replicated must be used.
Arguably the most common resource that is used is computation power, as in Bitcoin. 
Another common option, that we shall employ, is to use the coin itself.
In the following subsection we describe our permissionless approach for solving the cryptocurrency problem.

\subsection{Banking and Democracy}

%
Recall that the administrators are the ones that manage the consensus, concerning the existing balance and accepted transactions.
However, we claim that the consensus should be between the users.
I.e., it is the users who should be interested in the consensus. 
Without consensus there is no value to the coins they hold, as there is no agreement on the amount of coins each user holds.
The more coins you hold, the more responsibility you have for the cryptocurrency's future, as you will lose more if its value will drop.
The way we can coordinate between all the different users is by means of democracy -- letting the majority of the ``people'' decide.
However, the ``people'' in our case will not be the users, bur rather the coins.
As coins belong to users, we will provide each user with voting power that is proportional to the amount of coins he owns.
The problem is that the users are too numerous, and asking all of them to vote on each decision (i.e., transaction) is not practical.
Recall that in a democracy the people don't need to accept every rule, they just need to choose representatives to make the decisions for them (or on their behalf).
Thus, our users simply need to choose representatives. 
The most trivial representatives are the administrators, whose job is to maintain the consensus.
We shall now discuss the identities of the administrators.

In the real world we don't keep the money ourselves\hideC{ (at least most of our money)}.
Instead, we let someone -- the bank, to keep it for us.
When we want to use this money we simply address our bank\hideC{ (either directly or indirectly\footnote{Most often we address our credit card company, and they address 
our bank.})}.
In cryptocurrency, on the other hand, the amount of money you have is decided by agreement between the admins.
The advantage in such a scheme is that you are not depended on a single bank but rather on the agreement between multiple admins.
Moreover, the operations of the admins are completely transparent and, in fact, each of us can become an admin (or at least can gather the information that they get) and check for himself the correctness of the consensus.
Yet, there is some convenience in the banks scheme where each user has a single address to all of his requests\hideC{ (convenient both for the user and for the entire system)}.
Our solution merges the responsibilities of everyday banks and cryprocurrency admins by making the admins to function as banks. 
Each user will choose a bank (an administrator) where his money will be deposited, and whom he should directly address for every request.
Of course the user must also have the option to switch banks, so he won't lose his money in case his bank fails or otherwise ignores his requests. 
A bank, just like a user, will have a pair of secret and public keys.
The user's account number will be the combination of his bank's public key and his own public key.
\hideC{The bank's public key can be seen as the ``branch number'', while the user's public key is the ``inner account number'' in that specific bank.}

Let's assume I ask my bank to transfer money from my account to another.
As we are living in a democracy, the transfer must be accepted by the holders of the majority of the coin\hideC{ (i.e., a group of coin holders that together possess a majority of the money must accept this transfer)}.
In our approach, each bank represents its clients and `votes' on their behalf.
Roughly speaking, each bank gets its voting power according to the sum of money in the user accounts it manages. 
Once the banks agree on a transaction, we can see it as if the money holders themselves agreed on that transaction\hideC{ (as each money holder delegated his voting power to his bank)}.
{\quad }We shall now briefly describe the implemented protocol.
In blockchain based coins there is a single global ledger (the \mydef{blockchain}) that lists all the transactions, and all the admins (should) agree on this ledger.
In our system every admin shall have a private blockchain of its own, that lists (mostly) the transactions its clients committed.
That blockchain might also include a transaction of a user of another bank, in case that user asks to leave his original bank and move to this bank. 
If we want to compute the balance of a client, observing only the blockchain of his bank is not enough, as it doesn't list money that he receives from clients of other banks.
The information about such money transfers is found in the blockchains of the other banks (the banks whose clients transferred that money).
As banks must be able to compute such account balances, 
they must hold all of the existing blockchains.
In fact, each block in a given blockchain will contain pointers to blocks in other blockchains. 
More information appears in \Cref{sec:settings}.

\hideD{
So, when I send a transaction to my bank, my bank will check this transaction to make sure that I have the required money to spend, then it will probably group it with transactions of other users into a block, and it will chain this block to its private blockchain.
Next, my bank will send this new block to all the other banks, in order to receive their acceptance for the block.
Once a majority of banks (by means of voting power) sent their acceptance for the block, my bank has the proof that it should be accepted. 
More information concerning the exact protocol appears in \Cref{sec:prot}.}

The presented approach includes many advantages.
Among else it presents a completely decentralized and trustless system, an agreement mechanism without superfluous energy consumption and a deterministic cryptocurrency that operates over asynchronous channels and (theoretically) achieves low latency.
From the user side, just as in real life, you address all of your requests directly to your bank.
The banks have incentives to give their clients good service, for otherwise the clients will move to other banks (and less clients means less income).

\section{Related Works}\label{sec:related}
Maybe the greatest challenge of cryptocurrency is reaching consensus on the money balance and/or committed transactions.
The most prominent approach to this problem, introduced by Bitcoin \cite{bitcoin}, is the famous blockchain.
The difficult question is how a new block is added to a given chain, in such a way that everyone will agree that this block is indeed the next block in the chain.
Bitcoin's approach is to do so by a race based on computational power.
\extraB{That race causes a great energy waste.}\hideB{
The chance one has to win such a race is roughly equal to its computation power divided by the overall computation power that participate in this race. 
A winner of a race manages to create a block that will be accepted by everyone. 
}\hideD{If a node receives a block that is in the end of a chain that is longer than the chain it currently knows of, then it should forget the older chain, and adopt the new chain.
By assuming that most\footnote{
According to \cite{selfish}, even if most of the computation power is in good hands, it still might be not enough for achieving fair results.}
of the computational power is in ``good'' hands, that will follow the above protocol (i.e., adopting longer chains), it is argued that it will benefit everyone to follow that protocol as well.
Thus, if one node manages to create a new block and spread it to everyone before anyone else does, that block will be in consensus.
}This concept is called Proof-of-Work (PoW).
The great downside of PoW is its resource consumption.\footnote{It should be mentioned though that the requirement for resource consumption plays another important role -- it makes the option to alter history unlikely, as one has to invest more resources than all the resources that have been used, starting from the point of required change in history.}
In order to \hideB{alleviate the resource consumption}\extraB{avoid the energy waste}, the idea of Proof-of-Stake (PoS) \cite{pos} was introduced, where instead of deciding the next block by stochastic means, based on ownership of external resources (such as computation power), we can decide it based on inherent resources -- the coin itself.
The idea is that if you have more of the coin, you can get the chances to create more blocks. 
~ In our case we use the coin as a voting power, that is delegated by its holders to their banks, and where transactions are accepted by majority of the voting power. 
Delegating voting power is not new \cite{nano,dpos,neo}.
\hideC{The advantages of delegating the voting power is that we can apply consensus by `democracy', where the coin holders are effectively those that decide.
}Maybe the most prominent problem in delegating power is the \mydef{indifferent user}.
The problem is when users don't really care to whom they delegate their power.
The result might be that most of the honest users delegate their power to non functioning or malicious administrators\hideC{, or even don't delegate their power at all (in which case the rest of the users, many of whom might have malicious intentions, effectively have more power than their fair share)}.
In our approach, the users cannot truly be indifferent. 
Our users must choose a bank, and as they can get service only from that bank, they cannot delegate their voting power to a bank that is not functioning or that was proved to be malicious\hideC{ (a bank that acts maliciously will be ignored by the rest of the banks, so it won't be able to give services to its users)}.
%

Regardless the method been used to overcome the permissionless settings, using a blockchain results in relatively low throughput and high latency.
\hideB{
The reason for this problem is that, effectively, only one admin at a time can add a new block to the chain, which causes great congestion. 
}In order to solve this problem, one option is to do things outside the main chain (e.g.~\cite{lightning}), and update the main chain only when necessary\hideC{ (less frequent updates shall lead to less congestion)}.
Another option, however, is to completely avoid a main blockchain.
While blockchain has the convenient feature that it defines a total order over all the committed transactions, this convenience comes with a great cost, and is not truly necessary. 
\hideB{However, we still need some sort of partial order.
For example, if Alice wishes to transfer some money to Bob, we want to see that Alice received enough money \textbf{before} that, so she indeed has enough money in her account to perform that transfer.
}For keeping only partial order, we can use a DAG (directed acyclic graph). 
For example, each of~\cite{iota,nano,spectre,blockclique,hashgraph} uses some sort of a DAG (implicitly or explicitly). 
\hideB{

Instead of agreeing on a single blockchain, the admins will agree on this graph (the DAG).
}Each node in the \hideB{graph}\extraB{DAG} can contain user transactions, similar to a block in the blockchain.
%
%
Maybe the simplest use of DAG is as a set of parallel blockchains that are connected between them\extraB{ \cite{hashgraph,blockclique,unchain}}.
\hideB{For example, every admin can have a blockchain of its own (a chain of blocks that all of them belong to that admin), and every block in such a blockchain can reference blocks from blockchains of other admins (e.g., the ``gossip'' graph in \cite{hashgraph}).
In such case there is no competition between admins, and no congestion, as every admin can freely create new blocks in its own chain.
However, the problem is that conflicting transactions might appear on different blockchains, in blocks that have been created concurrently.
Conflicting transactions are transactions that cannot be applied together.
E.g., if Alice has \$40 and she issues two transactions in which she transfer \$20 to Bob and \$30 to Charlie, then clearly the two transactions cannot both take effect.
When using a single blockchain this was easy -- one of the two transactions had to appear before the other, so the second transaction would become invalid. 
But what should we do if these two transactions appear in two different blocks that were added concurrently to a graph at two different locations? 
Different DAG-based cryptocoins deal with such a problem in different ways.
A simple solution is that transactions of the same ``type'', that might be conflicting, are allowed to be introduced only at a specific location in the graph, so they cannot appear concurrently at different locations.
E.g., if our graph is indeed a set of parallel blockchains, then we can divide the transactions according to their source accounts, such that transactions of a specific account are allowed to appear only in blocks of a specific blockchain \cite{nano,blockclique,unchain,consensusNumber}.

In the above paragraph we mentioned two possible properties for a DAG that is constructed from parallel blockchains:
(1) That there will be a separate blockchain for every admin, so there will be no congestion, and (2) that conflicting transactions cannot be issued on different blockchains.
Combining both of these properties can have a good effect.
However, if even one admin stops responding, then certain transactions cannot be issued, which is clearly unacceptable.
Another option, of \cite{nano}, is that instead of one for every admin, there is a blockchain for every user.
This could be ideal, if we could trust all the users. 
However, as a malicious user can create conflicting blocks in its chain, there must be some sort of agreement on the actual blockchains.

Our solution maintains a blockchain for every admin(/bank), and users should submit transactions to the specific admin they chose (their bank), so that conflicting transactions are not supposed to appear in different blockchains.
However, the users can also choose to submit a transaction to another bank, when they want to transfer their account to that bank.
Thus, it is not really guaranteed that there won't be conflicting transactions (though, they should be less frequent -- only when users switch banks).
}

Arguably the strongest theoretical result that appears in this paper is that a deterministic cryptocurrency ststem can be implemented over asynchronous communications.
A similar idea appears also in \cite{consensusNumber}.
However, they assume that each account has a single process through which it submits its transactions (i.e., a single admin), which is not practical, as users won't be able use their money if the only admin that can publish their transactions has stopped responding.



\section{Model}\label{sec:model}
We define a cryptocurrency system by the tuple~$(A,I,\mathbb{B},P)$ where $A$ is a set of possible account numbers, $I:A\to\mathbb{R}^{\ge0}$ defines an initial balance for each account, $\mathbb{B}$ is a set of possible admin IDs and~$P$ is the protocol that every admin should follow.
$A$ and $\mathbb{B}$ might be infinite sets.
We assume that each account number has a matching private key that can be used to create digital signatures in a way that everyone can verify but no one can fake.
Accordingly, computing a public key from a private key should be easy but not vice versa. 
In permissioned settings, there would be $|\mathbb{B}|$ known nodes, that each will have a specific ID from~$\mathbb{B}$.
In permissionless settings, on the other hand, we can assume that $\mathbb{B}$, similarly to~$A$, represents a set of public keys, that each has a matching private key.
In this paper we focus on permissionless settings.

A transaction is defined by~$(s,t,m,i,d)$ where $s,t\in A$ are the source and destination accounts respectively, $m\in\mathbb{R}^+$ is the amount of money to be transferred from~$s$ to~$t$, $i\in\mathbb{N}$ is a sequence number and~$d$ is a digital signature of $(s,t,m,i)$ created using the private key of~$s$.
The first transaction issued by~$s$ should have the sequence number $i=1$, and the sequence number should be increased by 1 for every following transaction.
The sequence number has two uses: (1) It allows the creation of two identical money transfers\hideC{ (transferring the same amount to the same destination)}, and (2) It allows us to define below a simple requirements from the system.\footnote{Alternatively, we could employ Bitcoin's transaction mechanism, where each transaction consumes previous transactions.} 

In every execution of the system, there are a set of nodes that act as \mydef{users} and \mydef{admins}.
Every user knows the private key of one or more accounts.
Every admin has one or more distinct IDs from~$\mathbb{B}$. 
During the execution, the users create transactions and send them to the admins.
The admins can send messages between them and \mydef{accept} transactions or \mydef{reject} them.
Accepting a transaction means that the money transfer described by the transaction was executed (so, from now on, the money belongs to the recipient).
A transaction that was rejected might be later accepted, but not vice versa.
I.e., accepting a transaction is irreversible.
\hideC{This is important, as the recipient of the money might assume that he already got the money, and he might deliver some other good in return.
}By applying the set of accepted transactions we can compute how much money each user holds -- his initial balance according to~$I$ plus the sum of money he received minus the sum of money he transferred.

We use a permissionless settings, with asynchronous communications and byzantine crashes.
We say that an admin is \mydef{valid} if it follows the prescribed protocol and if it doesn't crash.
A formal definition for executions under these settings appears in \Cref{appendix:exec}.  
An execution is considered valid if it satisfies the following properties:
\begin{itemize}
\item \textbf{Agreement}: A transaction accepted by a valid admin will be eventually accepted by all valid admins.
\item \textbf{Positive Balance}: At every time point, and for every valid admin, applying the set of transactions that were accepted by that admin results in non-negative balance in all the accounts.
\item \textbf{Termination}: Every transaction that is sent from a user to a valid admin must be eventually either accepted or rejected by that admin.
\item \textbf{Rejection Restriction}: A valid admin may reject a transaction~$(s,t,m,i,d)$ only if either (1) there is an $i'\in\mathbb{N}$ such that $i'<i$ and the admin didn't accept any transaction of~$s$ with the sequence number~$i'$, (2) the user issued another transaction~$(s,t',m',i,d')$ where either $t\neq t'$ or $m\neq m'$, or (3) the balance of the user (computed by applying all the transactions that were accepted by this admin) is smaller than the amount to be transferred by the transaction.
\end{itemize}
The last requirement of rejection restriction is provided to avoid a trivial solution where all transactions are rejected.

\section{Solution}\label{sec:solution}
We now describe a cryptocurrency system whose executions satisfy the  requirements described in \Cref{sec:model} (under a specific assumption).
In our system, every admin and every user hold a pair of public and secret keys.
Denote by~$\mathbb{B}$ and~$\mathbb{C}$ the set of possible public keys for admins and users respectively (both might be infinite, possibly $\mathbb{B}=\mathbb{C}$).
We define $A=\mathbb{B}\times\mathbb{C}$ to be the set of possible account numbers.
I.e., every account number is a combination of a user's public key and an admin's public key.
%
To the admins we call \mydef{banks}.
An account number is said to belong to a specific user, under a specific bank.
A user should send his signed transactions to his own bank.
If he wants to switch banks, however, he should be able to submit to another bank -- his new bank, a transaction that transfers all his money to that bank.
\hideB{Moreover, there are cases where the user should be able to resubmit to another bank the same transaction that he has already submitted to his original bank.
An example is if the original bank of the user has stopped responding after receiving that transaction from the user, but before it managed to accept it.}
For simplicity, we assume that a user can submit every transaction to every bank of his choice.\footnote{A more practical option is that a bank will accept a transaction from an external account only if either (1) an account of it is the destination of the transaction, or (2) if that external user will attach another transaction that transfers money to an account in this bank. This way that bank has an incentive to accept that transaction.}

In order to list the accepted transactions, different cryptocurrencies use different types of ledgers that are maintained by the admins.
The most famous ledger is the blockchain.
On the contrary, we use a ledger in the form of a DAG.
We define that ledger in the subsection below, and in the subsection that follows we describe the protocol~$P$ that a bank (admin) shall follow. 

\subsection{The Blockgraph}\label{sec:settings}
We now define the \mydef{\Blockgraph} -- a DAG that we use instead of a blockchain.
In blockchain, every block contains a set of accepted transactions, and the different blocks are ``chained'' one after the other.
In our \blockgraph, every bank has such a chain of its own that contains only blocks that were issued by that bank (they must be digitally signed by it).
However, instead of blocks we talk about nodes, where each block from the original blockchain is split into several nodes. 
The beginning of a block is a \mydef{\startB} that contains the transactions to be committed.
A block ends at an \mydef{\acceptB}.
Between a pair of \startB\ and \acceptB\ there might be additional nodes, \mydef{\updateB s}, that contain references to nodes in the chains of other banks.
All the nodes between a pair of consecutive \startN\ and \acceptN\ nodes are considered as a single block. 
To define the initial coin distribution ($I$), we use an \mydef{\initB}.
The \initB\ contains a list of account numbers and a positive amount of money for each account in the list.
The list of accounts that appear in an \initB\ must be sorted, so that for every initial balance mapping~$I$, there will be a unique matching \initB.
See \Cref{fig:graph} for an example of a \blockgraph.\hideB{\\}
\extraB{A formal definition of a \blockgraph\ appears in~\Cref{appendix:blockgraph}.}

\hideB{\noindent }Each node in the graph (except the \initB) contains\hideC{ the following information:
\begin{itemize}
\item The bank number that owns that node.
\item The sequence number of the node in its bank's chain. 
\item The hash of the previous node in its bank's chain (its parent) or the hash of the \initB\ if it is the first node in the chain.
\item A digital signature of the owner bank.
\end{itemize}
}
\extraC{(1) the bank number that owns that node, (2) the sequence number of the node in its bank's chain, (3) the hash of the previous node in its bank's chain (its parent) or the hash of the \initB\ if it is the first node in the chain, and (4) a digital signature of the owner bank.
}
\hideC{
Excluding the \initB, there are three types of nodes.
Each node might contain additional information according to its type:
\begin{itemize}
\item A \startB\ contains user transactions the bank wants to apply.
\item An \updateB\ contains references to nodes of other banks.
\item An \acceptB\ contains no additional information.
\end{itemize}
}
See \Cref{fig:nodes} for an example of nodes.\hideC{\\}

\begin{figure}[t]
\begin{minipage}{.5\textwidth}
\ifthenelse{\boolean{showC}}
{\includegraphics[width=\textwidth]{\mainDirLoc 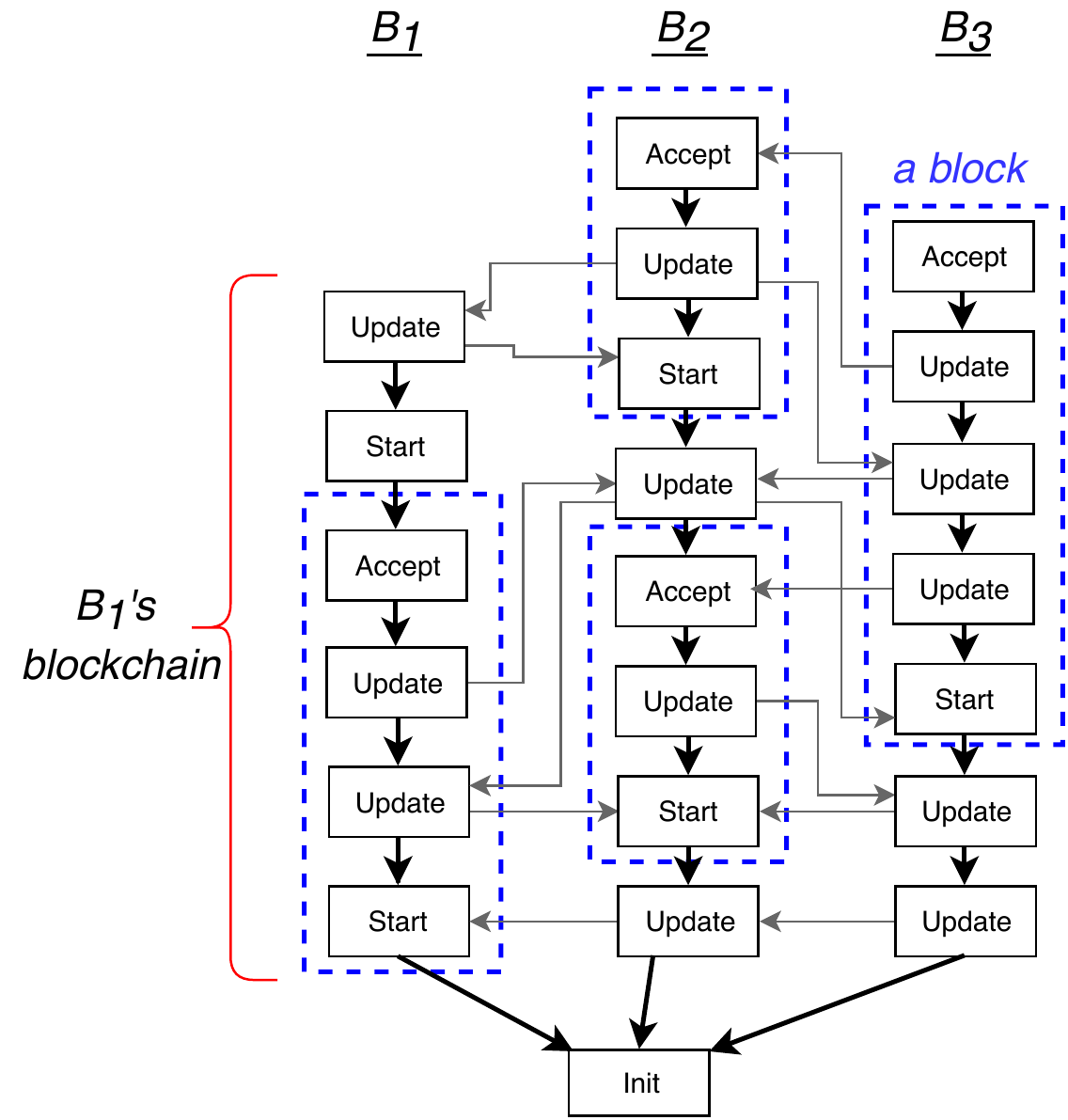}}
{\includegraphics[width=\textwidth]{\mainDirLoc 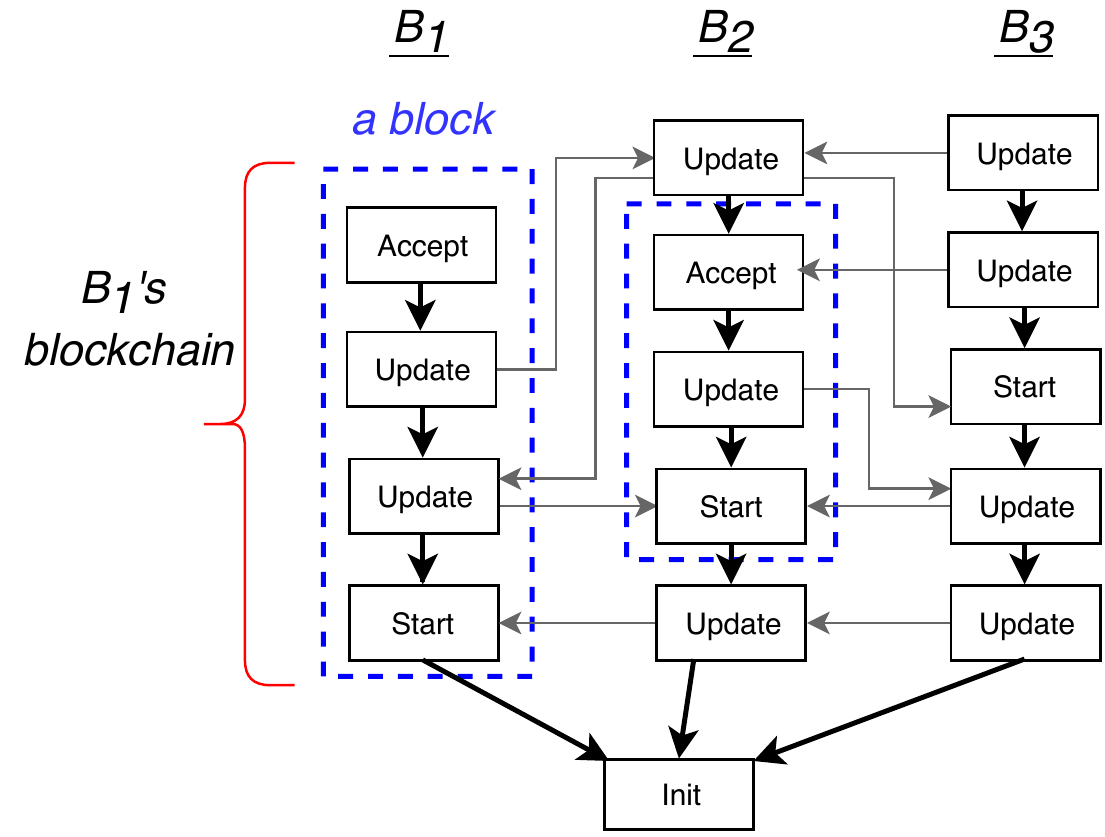}}
\caption{A \blockgraph\ with an \initB\ (bottom) and with ``blockchains'' of three banks: $B_1$, $B_2$ and $B_3$, where every such blockchain is constructed by \startB s, \updateB s and \acceptB s.
Edges from nodes to their parents appear in black, and edges from \updateB s to their referenced nodes in gray.
Note that the first node of every blockchain (the bottom node in every column) references the \initB\ as its parent.
Dotted blue rectangles mark complete blocks (a set of nodes from \startB\ to \acceptB).}
\label{fig:graph}
\end{minipage}
\begin{minipage}{.05\textwidth}
\hspace{.\textwidth}
\end{minipage}
\begin{minipage}{.45\textwidth}
\ifthenelse{\boolean{showC}}
{\includegraphics[height=18.1\baselineskip]{\mainDirLoc 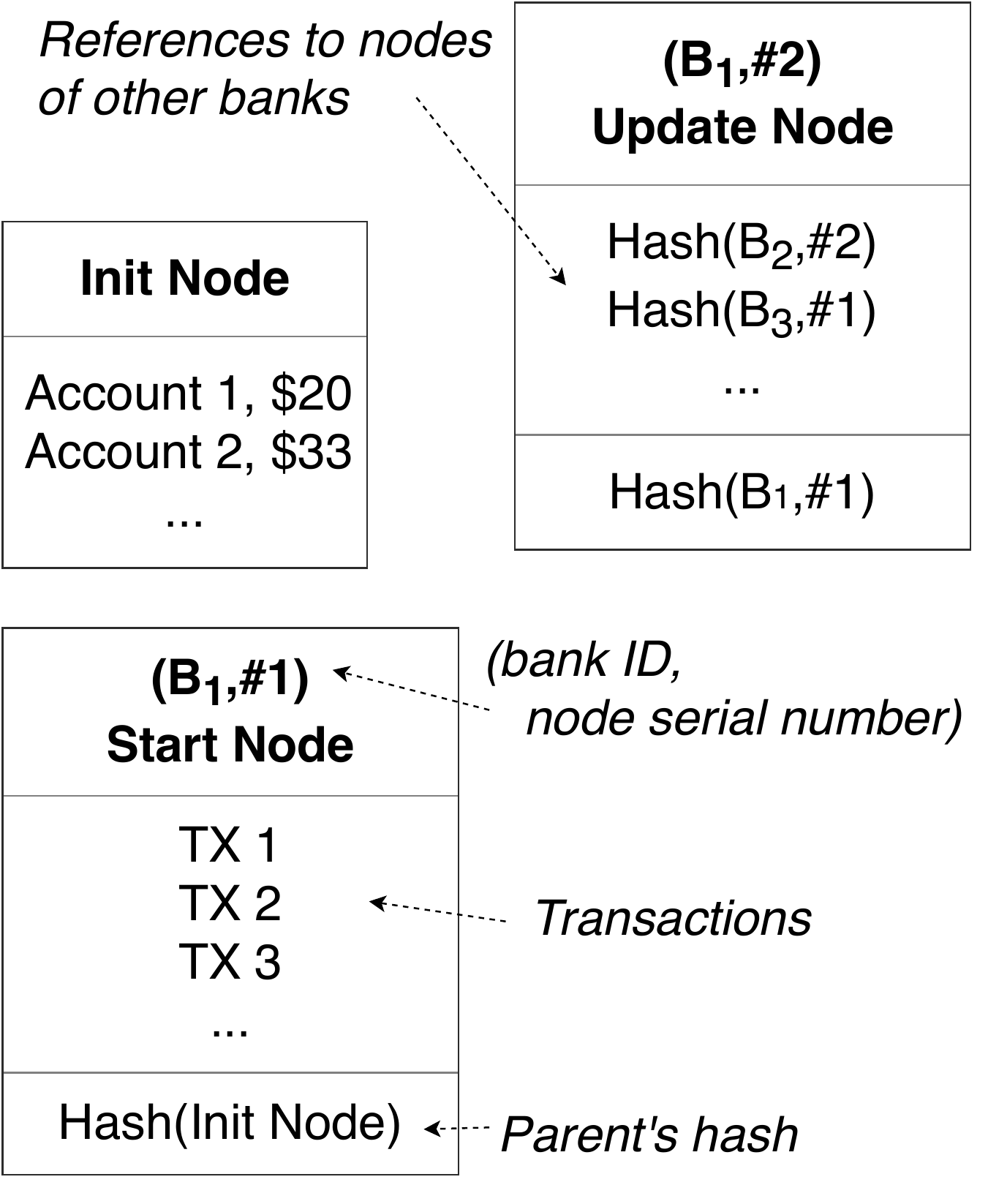}}
{\includegraphics[height=12.1\baselineskip]{\mainDirLoc Figures/Blocks.pdf}} 
\caption{Nodes example. 
An \initB\ (top left) describes an initial coin distribution, a \startB\ (bottom) lists issued transactions and an \updateB\ (top right) references nodes of other banks.
Note that the parent of the \startB\ (in this example) is the \initB, and the parent of the \updateB\ is the \startB.}
\label{fig:nodes}
\end{minipage}
\end{figure}

A bank creates a \startB\ when it wants to publish transactions it received from its users.
It then sends the issued \startB\ to the other banks.
The other banks in return can create \updateB s that reference that \startB.
A reference to a node is basically its hash together with its true id (the id of the bank that issued it, and that node's sequence number in its bank's chain).
When a bank creates such an \updateB, this means that it \mydef{acknowledges} the node it references.
Acknowledgment applies even with indirect references -- every node \mydef{acknowledges} all the nodes that are reachable from it (by means of paths in the graph). 
\hideC{For example, in \Cref{fig:graph}, the top \acceptB\ of~$B_\hideC{2}\extraC{1}$ acknowledges all the nodes in the graph excluding the two top nodes of~$B_3$.

}A bank that created a \startB\ has to wait for other banks to acknowledge its \startB, so that this node will be in ``consensus''.
To show that other banks have indeed acknowledged its node, it has to acknowledge their acknowledgments.
It does so by creating \updateB s of its own, that reference the \updateB s of the other banks (directly or indirectly).
Once enough banks have acknowledged its original \startB, and it has acknowledged those banks' acknowledgments, that bank can tell everyone to accept the transactions that are listed in this \startB.
It does so by creating an \mydef{\acceptB}, that marks the end of the block.
\hideC{This \acceptB\ doesn't need to contain any additional information, as it is a complementary node to the original \startB, and a new \startB\ cannot be issued before an \acceptB\ for the previous \startB\ has been created.}
\hideB{

We conclude this scheme by formally defining the structure of a \blockgraph:
\BlockgraphDefinition

Note that the last condition means that a given bank shall have no \acceptB\ without a matching \startB\ that comes before it, and no \startB\ that comes after another \startB\ without an \acceptB\ in the middle. \updateB s can appear everywhere.
A simple property of \blockgraph s that we shall use in the sequel is described in the following lemma:
\BlockgraphSubgraph
}

Before a bank accepts a node it received from another bank (by referencing it with an \updateB), it must check its validity.
Some of the conditions for validity are specific to the node's type, and some are general.
We start with the general conditions.
\hideC{

}Let~$v$ be a node. We say that a \blockgraph~$G$ \mydef{represents}~$v$ if~$v$ appears in~$G$, and if there are paths in~$G$ from~$v$ to every other node in~$G$.
Recall that~$v$ contains references to other nodes (their hashes) -- a reference to its parent, and possibly references to nodes of other banks, if~$v$ is an \updateB.
If we know those referenced nodes, then we can consider also the references that they contain, and so on, until we have a complete \blockgraph, which is exactly the representing graph of~$v$.
If~$v$ has no representing graph, then it is clearly not valid (it references, directly or indirectly, an impossible node).
If~$v$ does have a representing graph, we assume it is unique because
the references each node has to other nodes are defined by cryptographic hash, and we assume that it is impossible to create two different, valid, nodes with the same hash.
\hideC{Moreover, note that such two nodes must both belong to the same bank and be with the same sequence number, and possibly have additional requirements in order to fit in the right place in the graph.}\hideB{

}Considering the representing \blockgraph\ of~$v$, we remove~$v$ from this graph, and to the resulting graph we call the \mydef{subgraph of~$v$}.
\hideB{Note that the representing graph contains no cycles (as it is a \blockgraph), so there could be no edge from another node to~$v$ in that graph.
Thus, the conditions of \cref{lem:subgraph} applies for the subgraph of~$v$, and so it is also a \blockgraph.
}The subgraph of $v$ is\extraB{
also a \blockgraph\ (\Cref{appendix:subgraph}) and is}
considered valid only if all of its nodes are valid.
If the subgraph of~$v$ is not valid, then~$v$ is not valid.
%

Assuming that the subgraph of~$v$ is valid, we check $v$'s contents according to its type.
If~$v$ is an \initB, then it is trivially valid.
\extraB{The same goes with \updateB s (for now).
}\hideB{If~$v$ is an \updateB, then it could also be trivially valid, but in order to avoid waste of information we should make sure that an \updateB\ doesn't reference nodes that are already appear in its parent's subgraph.
Later on in this section, we introduce an additional requirement from \updateB s that intends to deal with malicious banks.
}\hideC{

}Before considering the requirements of \startB s, we define the way we compute the balance of each account according to a given \blockgraph.
Recall that the \blockgraph\ contains an \initB\ that encodes the initial balance, and transactions that appear in \startB s.
When computing the total balance, we consider only accepted transactions, i.e., transactions that appear in \startB s that have matching \acceptB s.
We compute the total balance by starting with the initial balance, and then applying the accepted transactions.
Note however that the same transaction might appear in more than one \startB\ (in \startB s of different banks), so we must make sure that we apply it only once.
Moreover, there might be \textbf{different} transactions of the same user with the same sequence number that appear in different \startB s of different banks.
Such transactions are conflicting, and we don't want to apply them together.
Recall that the set of nodes between a \startB\ and a matching \acceptB\ (including the \startB\ and the \acceptB\ themselves) are considered a single \mydef{block}.
We say that one block,~$b_1$, acknowledges another block,~$b_2$, if the \acceptB\ of~$b_1$ acknowledges the \startB\ of~$b_2$\hideC{ (i.e., the subgraph of the \acceptB\ of~$b_1$ contains the \startB\ of~$b_2$)}.\hideC{
For example, in \Cref{fig:graph}, the bottom block of~$B_2$ and the block of~$B_1$ both acknowledge each other (note the \updateB\ that appears below that block of~$B_2$).
On the contrary, the block of~$B_3$ acknowledges that bottom block of~$B_2$
, but not vice versa.

}
Assume that there is a pair of conflicting transactions that appear in two blocks $b_1$ and $b_2$.
If $b_1$ acknowledges $b_2$, then we don't apply the conflicting transaction from $b_1$ 
as it ``knows'', at the moment of accepting its transactions, that there is already another conflicting transaction. 
If both $b_1$ and $b_2$ acknowledge each other, then both transactions won't be applied.
The problem is if no block acknowledges the other, as both conflicting transactions will be applied.
\extraB{We shall consider that case later.}\hideB{
We shall soon consider that case, but we first conclude the total balance computation:
\begin{definition}[Total balance]
Given a \blockgraph, we take the initial balance according to the \initB, and then for every \acceptB\ we apply every transaction that appears in its matching \startB, as long as the \acceptB's subgraph doesn't contain a conflicting transaction, and as long as the same transaction wasn't already applied earlier in the computation.
\end{definition}}

We can now define the validity requirements for \startB s.
Let~$v$ be a \startB\ that belongs to bank~$B$ and whose subgraph,~$G$, is valid.
$v$~will be valid if all the transactions that it contains are valid.
Let~$t$ be a transaction that appears in~$v$, that belongs to the user~$u$ and transfers amount of~$m$ to another account.
Roughly speaking, $t$~is considered valid if it is indeed the next transaction in the transaction-chain of~$u$, and if~$u$ has a balance of at least~$m$.
We start with deciding if~$t$ is indeed the next transaction in~$u$'s chain.
First, we consider all the accepted transactions of~$u$ in~$G$ (ignoring identical transactions).
Assuming there are~$N$ such transactions, their sequence numbers must cover all the range of 1 to~$N$, or else~$t$ will be considered invalid, as that user has no valid transaction chain.
The sequence number of~$t$ must be~$N+1$ in this case, and there must be no other, different, transaction of~$u$ with the sequence number of~$N+1$ that appears in~$G$.\hideB{ There might be however another transaction that is identical to~$t$, as long that it appears in a \startB\ of another bank,~$B'$ (where $B'\neq B$).
In this case this means that~$u$ has resubmitted~$t$ to~$B$, maybe because~$B'$ stopped responding. 
Note that there cannot be a matching \acceptB\ for that \startB\ of~$B'$ (in~$G$), because then this identical transaction would either be accepted, and then~$t$ is not valid because it should have a bigger sequence number, or otherwise it would be rejected, which means that there is another, different, transaction with the same sequence number ($N+1$), so~$t$ is, once again, invalid.}
\extraB{Next, we check the balance of~$u$ by computing the total balance according to~$G$ (the subgraph of the \startB\ $v$), and make sure that it has a balance of at least~$m$.
}\hideB{The second requirement for~$t$'s validity is that~$u$ has a balance of at least~$m$. We decide the balance of~$u$ by computing the total balance according to~$G$ (the subgraph of the \startB\ $v$).
}If~$t$ fulfills both requirements it is considered valid.

Next, we need to define the validity requirements of an \acceptB.
Let~$v_a$ be an \acceptB\ that belongs to bank~$B$, 
and let~$v_s$ be the \startB\ that comes before~$v_a$ in~$B$'s chain.
The creation of~$v_a$ signals that the transactions of~$v_s$ should be accepted.
$B$ can create~$v_a$ only once that~$v_s$ is in consensus.
Consensus in our system is defined by coin possession.
Each user has a voting power that is proportional to his balance, and that voting power is delegated to the user's bank.
$v_a$ is considered valid if it is evident from its subgraph that there is a set of banks, $S$, that together hold voting power above a predefined threshold, and that they all acknowledged~$v_s$.
\hideB{More precisely, a bank $B'$ is in~$S$ if (and only if) it has a node~$v$ such that (1) $v_a$ acknowledged $v$ and (2) $v$ acknowledged $v_s$.
Recall that a node~$v_1$ acknowledges another node,~$v_2$, if $v_2$ appears in~$v_1$'s subgraph.
Recall also that the sequence of nodes of~$B$ starting from~$v_s$ and ending at~$v_a$ are considered a single block (see \cref{fig:graph}).
If $B'\in S$ we say that~$B'$ \mydef{supports}~$v_a$'s block.
}To finish the definition of validity requirements of an \acceptB, we need to define how to compute the voting power of each bank\hideB{ and to define the threshold of required supporting voting power}.
We shall soon do so.
%

When computing the total balance, we mentioned above that there is a problem if two blocks that contain conflicting transactions don't know of each other.
In such case, the two conflicting transactions will be both applied when computing the total balance, and that might result in an account that has a negative balance.
\hideC{Such account has managed to spend more than what it originally had, and unlike in real life, no one can come to that account owner and ask the money back.
}To avoid this problem, we want to make sure that all the blocks know each other. 
This brings us to the following definition.
\begin{definition}[Proper graph]
We say that a \blockgraph\ is \mydef{proper} if for every pair of \acceptB s, at least one of them acknowledges the \startB\ of the other. 
\end{definition}
\hideB{We shall prove later that computing
}\extraB{Computing }the total balance of a \blockgraph\ that is proper and valid results in non-negative balance for all the accounts.
Thus, we have big interest in proper graphs.
\extraB{Roughly speaking, we can assure that the graphs will be proper by requiring that every \acceptB\ will be supported by majority of the voting power.}\hideB{ ~
How can we make sure that a (valid) \blockgraph\ will remain proper?
I.e., that for every two \acceptB s, at least one of them acknowledges the \startB\ of the other?
Roughly speaking, we can get this effect by making sure that every block (/\acceptB) is supported by the voting power of majority of the banks.
In this case, every two different blocks (/\acceptB s) will be obliged to share some common supporting voting power.
I.e., there must be at least one bank that supports both blocks.
Let~$B$ be such bank, that supports two blocks, $b_1$ and $b_2$, and assume it first acknowledges $b_1$ and then acknowledges $b_2$.
When $b_2$ acknowledges $B$'s acknowledgment on $b_2$ itself, it acknowledges also the acknowledgment of~$B$ on~$b_1$, and so it indirectly acknowledges~$b_1$.
As a result, the two blocks are indeed connected.

We shall now define how to compute the voting power distribution.
}\extraB{~ ~}Every \blockgraph\ defines a specific voting power distribution between the banks.
In its most basic form, the voting power of a bank is the amount of money belong to the users of this bank.
Thus, to compute the voting power distribution, we start by computing the total balance according to the graph.
Recall that computing the total balance considers only accepted blocks (i.e., completed blocks -- \startB s that are followed by matching \acceptB s).
This is not enough, however.
We might have a problem with ``open blocks''~--~\startB s that don't have a matching \acceptB.
If a transaction in such an open block transfers money from one bank to another (or, more precisely, between users of these banks), then the transferred money is ``in transit'' between the two banks.
We name the bank from which money is transferred as the ``source bank'', and the bank that receives the money as the ``destination bank''.
The true concern is when such a transaction appears in a \startB\ of a bank which is not the source bank. 
In such case, some banks might see extensions of this graph, where that open block is already closed, and the money no longer belongs to the source bank, but to the destination bank. 
Yet, the source bank might not be aware of this ``closure'', and it might keep using that money as part of its own voting power.
The result is that the same money(/voting power) might be used to support different blocks, that might not be aware of each other, and this might result in a non proper graph.
To avoid this problem, we define such money to be \mydef{shared} between the two banks (the source and destination banks).
A shared money can support a block only if all the banks that share it support that block.
\hideB{

An exception is the case where such a problematic transaction, that appears in an ``open'' block, is doomed to be rejected because there is already another (or identical) transaction of the same user with the same sequence number that is accepted according to this graph.
On the contrary, there might be several conflicting problematic transactions, in different open blocks. 
In this case, the money will be shared between the source bank and all possible destination banks.}
More formally we define the following set of problematic transactions:
\begin{definition}[Uncertain transactions set]
Let~$G$ be a \blockgraph.
We define the set of \mydef{uncertain transactions} in~$G$ as follows.
We start with the set of transactions that appear in \startB s that don't have matching \acceptB s in~$G$.
Let~$t$ be a transaction from this set, that appears in an ``open'' \startB\ that belongs to bank~$A$ and that transfers money from an account in bank~$B$ to an account in bank~$C$.
The set of uncertain transactions includes~$t$ exactly if: 
(1) $A\neq B\neq C$ (possibly $A=C$), and
(2) there is no other transaction of the same user with the same sequence number as~$t$ that is applied when computing the total balance according to~$G$.
\end{definition}
\newcommand{\footnoteUncertain}{\footnote{A more accurate way to define the shared voting power is as follows.
First, find the transaction that wishes to transfer the lowest amount of money, and define its money has shared between all the banks from the uncertain transactions set of the user.
Then, find the transaction with the next lowest amount of money, and define the difference between the two transactions as shared between all the banks from the set above, removing the previous transaction.
Keep doing so until no transactions left.
When using this method, each bank will have a ``shared'' part that is equal to the amount that it can actual get, and no more than it.}}
\hideB{
The money that is transferred by uncertain transactions is exactly the money that we cannot be sure to which bank its corresponding voting power should be delegated -- the source bank or the destination bank.
Thus, we will define that money (or, more precisely, the corresponding voting power) as shared between the two banks.
As we shall now prove, if a user has several uncertain transactions, then they all have the same sequence number.
If our graph is proper (and remains proper), this means that at most one of these transactions might be later accepted, and so that money will belong either to only one of the destination banks (if one of the transactions gets accepted) 
or to the source bank (in case all the transactions get rejected).
Thus, we define that money as shared between the source bank and the destination banks of all those conflicting transactions (as exactly one of them might get that money).
Note that those conflicting transactions might wish to transfer different amounts of money, so we simply take the biggest amount among those amounts of money.
\footnoteUncertain\ 
We shall now prove what we claimed above:
\begin{lemma}
Let~$G$ be a \blockgraph\ whose \startB s are all valid, and let $t_1$ and $t_2$ be two transactions that appear in~$G$.
If $t_1$ and $t_2$ are uncertain transactions of the same user, then both have the same sequence number.
\end{lemma}
\begin{proof}
Let~$t_1$ and~$t_2$ be two uncertain transactions in~$G$ that belong to the same user,~$u$.
Assume by way of contradiction that they have different sequence numbers.
Without loss of generality, assume that~$t_1$ has a lower sequence number.
By the assumptions, the \startB\ where~$t_2$ appears must be valid.
According to the validity requirements of a \startB, there are in its subgraph (thus, in~$G$) accepted transactions of~$u$ with all the sequence numbers between 1 and one below the sequence number of~$t_2$.
Thus, there is an accepted transaction of~$u$ with the same sequence number as~$t_1$, and so~$t_1$ cannot be an uncertain transaction (according to the uncertain transactions definition).
%
\end{proof}
%
}
We now formally
define the voting power distribution: 
\begin{definition}[Voting power distribution]
Let~$G$ be a \blockgraph.
We define the voting power distribution of~$G$ as follows.
We start by computing the total balance, according to~$G$.
The initial voting power of each bank is the sum of money of its users.\hideC{

}
Next, we consider the uncertain transactions in~$G$, and divide them according to their issuing user. 
For each user we take the transaction that wishes to transfer the biggest sum of money.
That sum of money is decreased from the voting power of his bank, and is granted to a coalition that consists of his bank and the banks of all the users that might get money according to that user's set of uncertain transactions.
\end{definition}
We can now finish defining the requirements for an \acceptB,~$v_a$, to be valid.
Recall that we defined~$S$ to be the set of banks that support~$v_a$\hideC{ (where a bank~$B'$ supports~$v_a$ exactly if~$v_a$ acknowledges a node of~$B'$ that acknowledges~$v_s$ -- the \startB\ that comes before~$v_a$)}. 
The sum of voting power of banks from~$S$, including voting power that is shared exclusively by members of~$S$, must be above a predefined threshold that we shall soon consider.\hideC{
As we shall see, this threshold depends on what we assume about the banks nature.}

After defining the validity requirements of nodes, we say that a \blockgraph\ is valid if all of its nodes are valid.
We wish that a valid \blockgraph\ will be also proper.
If the banks are not malicious, then the voting power threshold required for an \acceptB\ to be valid can be defined as a simple majority (i.e., half the total voting power), and then we can prove that every valid graph is also proper.
However, an assumption that there are no malicious banks is not practical.
\hideC{We now extend our definitions in order to deal with malicious banks, and formally prove that valid \blockgraph s will be proper.
}A bank is considered malicious if it creates conflicting nodes -- different nodes with the same sequence number.
\hideC{The problem with such conflicting nodes is that they might seem valid for themselves (when observing their subgraphs), and different banks might adopt a different node of the conflicting pair.
Only when the valid banks find out that there are two conflicting nodes, they know that the bank that created them is malicious.
More precisely, if}\extraC{If} a \blockgraph\ contains conflicting nodes of a given bank, then this bank is considered malicious according to this \blockgraph.
\hideC{I.e., the \blockgraph\ itself holds the proof that the bank is malicious.

}We deal with malicious banks by defining \updateB s to be invalid if it is evident from their subgraph that their own bank is malicious.
I.e., a malicious bank cannot reference nodes of banks that already know that it is malicious, as it will make its \updateB s invalid and the other banks won't accept such invalid nodes.
\hideB{Another optional, cosmetic, requirement (which has no practical importance) is that \updateB s (of valid banks) won't reference nodes of banks that are already known to be malicious.
More precisely, if it is evident from the subgraph of the parent of an \updateB\ (the node that comes before that \updateB\ in its bank's chain) that a specific bank is malicious, then this \updateB\ should not directly reference a new node of that bank.}

Recall that we want graphs to be proper, so that conflicting transactions won't be accepted together.
With malicious banks around, we must make sure they don't turn our graphs to non proper.
The ``easiest'' way for a malicious bank to turn a graph to be non proper is by issuing two conflicting \startB s, each accompanied with an \acceptB\ of its own.
However, a bank cannot create a (valid) \acceptB\ at its will, as it needs the support of the majority of the voting power (above the predefined threshold).
We claim that the voting power of a non malicious bank can support at most one of such conflicting \acceptB s (will be proved later on).
% ...
Thus, by defining the threshold of the required voting power to be high enough, and by assuming that the voting power in the hands of malicious banks is limited, we prevent malicious banks the option of creating such two conflicting \acceptB s (that are based on different \startB s). \hideB{

}But, what prevents a malicious bank from creating two conflicting \acceptB s for a single \startB?
I.e., after gathering the required support for a valid, single, \startB, a malicious bank might create two conflicting, and valid, \acceptB s.
\extraB{The problem is if one of the transactions in the \startB\ will be accepted according to one \acceptB, but rejected as conflicting according to the other.
}\hideB{Assume that the malicious bank~$\bar{B}$ indeed creates such two conflicting \acceptB s $v_1$ and $v_2$ for a single \startB~$v$.
Moreover, assume that there is a transaction~$t$ that appears in~$v$ that is accepted at~$v_1$ but rejected at~$v_2$ (because there is an \updateB\ that comes before~$v_2$, but not before $v_1$, that references another block that contains a transaction that is conflicting with~$t$).
The result is that some of the banks might see~$v_1$, so they shall think that~$t$ should be applied, while other might see~$v_2$ and think that~$t$ should be rejected.
The problem in the above case is not the temporary situation where some banks accept an extra transaction~($t$) while others don't, as the latter banks shall sooner or later see also~$v_1$, that accepts~$t$, and then they should apply it as well.
Instead, the true problem is that this state can lead to creation of a valid graph that is not proper.
The reason is that~$v_1$ and~$v_2$ induce different voting power distributions.
Considering the voting power that corresponds to the amount of money that is transferred by~$t$, that voting power belongs to the destination bank, according to~$v_1$, and to the original bank according to~$v_2$.
Thus, each of the banks might think that this voting power is in its own hands, and each of them might use it to support a different block.
As a result, two different blocks might enjoy the support of the same voting power without knowing of each other, which might produce a non proper graph.

}
In order to solve this problem we separate the \textbf{actual} acceptance/rejection of a transaction (at an \acceptB) from the decision of whether it \textbf{would} be accepted or rejected. 
We introduce a new node to the \blockgraph\ -- a \mydef{\closeB}, that must appear exactly once between every pair of \startB\ and \acceptB.
The validity requirement of a \closeB\ is exactly as the requirement of an \acceptB, i.e., it must have a supporting voting power above a predefined threshold.
\hideB{More precisely, let~$v_s$ be a \startB, and let~$v_c$ and~$v_a$ be the \closeN\ and \acceptN\ nodes respectively that come after~$v_s$.
We define two set of banks,~$S_c$ and~$S_a$.
A bank $B'$ is in~$S_c$ if and only if it has a node~$v$ such that (1) $v_c$ has acknowledged $v$ and (2) $v$ has acknowledged $v_s$.
A bank $B''$ is in~$S_a$ if and only if it has a node~$v'$ such that (1) $v_a$ has acknowledged $v'$ and (2) $v'$ has acknowledged $v_c$.
$v_c$ will be valid if the sum of voting power of the banks in~$S_c$, computed according to~$v_c$'s subgraph, is above a predefined threshold, and $v_a$ will be valid under the same condition when considering the set of~$S_a$, and the voting power distribution according to~$v_a$'s subgraph.

}Using the \closeB, a malicious bank might still create two conflicting \acceptB s, but both of them will be based on the same \closeB\extraB{, and so the same set of transactions will be accepted or rejected in both.}\hideB{ (just as before that, two conflicting \acceptB s had to be based on the same \startB, or otherwise they wouldn't get the required supporting voting power).
Next, we change the definition of the total balance computation, such that a transaction is considered conflicting (and get rejected) only if it is evident from the corresponding \closeB\ (instead of the \acceptB):
\begin{definition}[Total balance \#2]
Given a \blockgraph, we take the initial balance according to the \initB, and then for every \acceptB\ we apply every transaction that appears in its matching \startB, as long that the subgraph of its matching \closeB\ doesn't contain a conflicting transaction, and as long that the same transaction wasn't already applied before that.
\end{definition}
Now, two conflicting \acceptB s that are based on a single \closeB\ have no effect.
The same transactions will be accepted or rejected in both.}
Note that a malicious bank might also create conflicting \closeB s (based on a single \startB), but as long that it will be able to complete only one of them with a matching \acceptB, there will be no effect.
This motivates updating the definition of proper graph:
\begin{definition}[Proper graph \#2]
We say that a \blockgraph\ is \mydef{proper} if 
(1) for every pair of \closeB s, at least one of them acknowledges the \startB\ of the other, and (2) for every pair of \acceptB s, at least one of them acknowledges the \closeB\ of the other.
\end{definition}
The first requirement, concerning \closeB s, is equal to the original requirement from proper graph, that makes sure that there is a connection between every two blocks, so that conflicting transactions won't be applied together.
The second requirement, concerning \acceptB s, comes to deal with malicious banks.
\hideC{I.e., it makes sure that malicious banks don't creates distinct \acceptB s to distinct (conflicting) \closeB s.}
We can now have the following claim (proved in \Cref{appendix:positive}):
\begin{lemma}
\label{lem:positive}
Computing the total balance of a \blockgraph\ that is valid and proper results in non-negative balance for all the accounts.
\end{lemma}
\hideC{
We want graphs to always be proper, where conflicting transactions cannot be accepted together 
and where balances are kept non negative.}
As we shall now claim, valid graphs will be proper indeed, assuming the voting power of malicious banks is limited.
In classical distributed systems theory we assume a fixed set of servers or nodes.
Dealing with malicious (more commonly called \mydef{byzantine}) nodes often involves an assumption that more than two thirds of the nodes are non-malicious (or, non byzantine).
This is similar in our case, only that there is no importance here to the number of the nodes (i.e., banks), but rather to the voting power distribution.
Voting power is considered ``valid'' if it belongs to a non-malicious bank, or shared by only non-malicious banks.
The voting power distribution should always be such that more than two thirds of the total voting power is valid.
Following that assumption, we define the threshold of required voting power for creation of valid \acceptN\ and \closeN\ nodes to be at least two thirds of the total voting power.\ 
\hideB{

}We are now ready for the main claim of this subsection\extraB{ (the proof appears in \cite{FullPaper})}\hideB{ (proved in \Cref{appendix:proper})}:
\begin{lemma}[Proper graphs]
\label{lem:proper}
Let~$G$ be a valid and proper \blockgraph.
If more than two thirds of the voting power in every subgraph of~$G$ is in the hands of banks that are not malicious according to~$G$, 
then extending~$G$ by a single valid node will result in valid and proper graph. 
\end{lemma}
%

\subsection{The protocol}\label{sec:prot}
In the previous subsection we described the \blockgraph\ that we use as a ledger.
We now define the protocol~$P$ that the banks should follow.
Every bank that takes part in the system must start with a \blockgraph\ that contains only the \initB\ that matches~$I$.
As time passes, each bank might add to its \blockgraph\ new nodes that it creates or receives from other banks.
\extraB{A basic protocol for the banks is to create valid \startN, \closeN\ and \acceptB s whenever they can.
If they receive nodes they from other banks, they must make sure that they are valid and that they don't make their graph improper, and then they introduce those nodes to their graphs and create corresponding \updateB s.
For every \acceptB\ that they introduce to their graph, they should \mydef{accept} the relevant transactions.
}\hideB{The protocol we use is rather straightforward: 
\Protocol
}
\extraB{Under the assumption that there is a finite set of non malicious banks that follow the above protocol and that together they always hold more than two thirds of the total voting power, we can prove that all the requirements that were defined in \sec{model} will be fulfilled. 
More information and a complete protocol and proof appear in \Cref{appendix:protocol}.}

\section{Conclusions}\label{sec:conclusions}
%
%
Cryptocurrency systems should allow users to easily transfer money between them, without excess bureaucracy and supervision.
The complexity of creating such systems is in managing the required consensus, regarding the accepted transactions, in a permissionless settings where every one can take part in the consensus process.
In \Cref{sec:model} we defined requirements that should form a cryptocurrency system that is ``fair'' to its users:
A user that creates valid sequence of transactions can be sure that its transactions will be accepted.
Moreover, the transactions, once accepted, will remain accepted, so the recipient can be sure that he got the money\extraB{  (this is not the case when using Blockchain)}.
\hideB{This requirement doesn't apply with blockchain, where a sequence of blocks might be accepted as part of the blockchain, and later be ignored, in favor of another, parallel, longer sequence.
}On the other hand, the requirements we define are flexible enough so that we can implement a system that satisfies the requirements under completely asynchronous networks. 
This is achieved by allowing the rejection of conflicting transactions without forcing acceptance of one of them. 

As a solution for the problem described in \Cref{sec:model}, we introduce a system 
that has beneficial characteristics of its own.
The system bases the consensus on the coin possession, so we avoid the great waste of energy that exists in proof of work coins.
Moreover, the latency in the system is only due to message transmission times and the time it takes to perform the basic required computations of validating/applying signatures.

Security in our case is very plain and simple.
By employing a deterministic protocol, that works even when communication is asynchronous, the only way to undermine the security is if more than one third of the voting power will be in malicious hands, or, more precisely, in your hands (it doesn't help if the other malicious banks don't cooperate with you).
%
%
%
% ...
\hideB{\subsection{Future Work}
There are yet many open challenges for implementing the system presented in this paper, and further research is required.
For example, we might want to spare memory, and not to remember the entire \blockgraph. 
The question is what can we forget, and under what conditions.
Another challenge is with large number of banks.
The more banks we have, the bigger the requirement for memory and bandwidth.
At some point, if there will be banks that won't have strong enough hardware, and their voting power on the other hand will be significant (so they are required for the consensus process to complete), it might cause extra delay in the system.

An existing downside in many cryptocurrencies including Bitcoin and the coin presented in this paper is that the system is completely transparent.
I.e., every transaction that is accepted is visible to the public, so that everyone can see the source and destination accounts, and the sum of money that was transferred.
Yet, the identities of the account owners are of course generally unknown. 
An interesting direction in our case is if we can limit the transparency to the level of banks, so that the only entities that truly know the transaction source, destination and sum of money (all together) are the banks of the issuing and receiving clients.

An important concept in cryptocurrencies at which consensus is based on coin possession, such as in Proof of Stake, is that malicious entities will be financially damaged because of their acts.
Otherwise, different entities (and especially strong entities, that posses a lot of coin) might maliciously attempt to maneuver the currency, without getting harmed.
This is the famous \mydef{nothing at stake} problem.
Note that this is never true in practice, because an unstable coin will have a lower value, so such entities that posses a lot of coin have an interest in keeping the coin credible.
Yet, if the coin value is decreased, this will harm also non malicious parties.
Thus, the best practice is indeed financially damaging those malicious entities.

Note that in our case banks should have a private account, to where they receive the commission on the users' accepted transactions.
A best practice will be to define that transactions that wish to transfer private money of a bank will be allowed to appear only in \startB s of that banks itself.
Recall that a malicious bank (that made a malicious act) won't mange to create new blocks, as it cannot reference nodes of other banks that are already aware to its malice so it cannot gain the required support for a block.
Thus, a malicious bank won't be able to use its private money (and once it turned out to be malicious, it won't receive any more commission).
If we can incentivize a bank to keep money in its private account, this can be used as means to make sure the bank follows the protocol, as if it won't, then it will lose that money.
This will be similar to real banks, that by regulatory requirements, in order to ensure their stability, must have some capital of their own (such stability is not relevant to our banks, as they cannot use their clients' money, as opposed to real world banks).

The most reasonable way to incentivize the banks to keep private money is by limiting their voting power in case they don't have enough private money.
I.e., we can define a percentage of private money out of the voting power of a bank that the bank must have in order to use its full voting power.
If it doesn't have the required percentage, its voting power will be decreased proportionally.
A more research is required in order to make sure that the system still complies with the agreement and termination requirements.
}


\bibliographystyle{abbrv}
\bibliography{Crypto}

\appendix

\section{Execution}
\label{appendix:exec}
\Execution

\section{Proof of \Cref{lem:positive}}
\label{appendix:positive}
\LemmaPositive

\section{Proof of \Cref{lem:proper}}
\label{appendix:proper}
\LemmaProper

\end{document}